\documentclass{layout}
\usepackage{amssymb}
\usepackage{graphicx}
\usepackage[colorlinks=true, linkcolor=blue, citecolor=webgreen]{hyperref}
\usepackage{color}
\usepackage{bbm}
\usepackage{eurosym}
\usepackage[round]{natbib}
\usepackage{IEEEtrantools} 
\newcommand{\R}{\mathbb{R}}
\newcommand{\E}{\mathbb{E}}
\newcommand{\G}{\mathbb{G}}
\newcommand{\N}{\mathbb{N}}

\newcommand{\A}{\mathcal{A}}
\newcommand{\C}{\mathcal{C}}
\renewcommand{\S}{\mathcal{S}}
\newcommand{\F}{\mathcal{F}}
\newcommand{\cR}{\mathcal{R}}

\newcommand{\B}{\mathcal{B}}
\newcommand{\cL}{\mathcal{L}}

\renewcommand{\L}{\mathrm{L}}

\newcommand{\cC}{\mathcal{C}}
\renewcommand{\P}{\mathbb{P}}
\newcommand{\D}{\mathcal{D}}
\newcommand{\X}{\mathcal{X}}
\newcommand{\cK}{\mathcal{K}}

\newcommand{\e}{\varepsilon}
\newcommand{\J}{\mathrm{L^{1}}}
\newcommand{\K}{\mathrm{L^{\infty}}}
\newcommand{\pas}{\P\text{-a.s.}}
\newcommand{\pae}{\P\text{-a.e.}}
\renewcommand{\d}{\;dt}
\newcommand{\q}{\quad\pas}

\newcommand{\xqed}[1]{%
  \leavevmode\unskip\penalty9999 \hbox{}\nobreak\hfill
  \quad\hbox{\ensuremath{#1}}}
  \renewcommand\qedsymbol{$\blacksquare$}
  \def\ignore#1{}
\def\b1{\mathbbm 1}
\def\Om{\Omega}
\def\om{\omega}
\DeclareMathOperator*{\ess}{ess\,sup}
\newtheorem{Theo}{Theorem}[section]
\newtheorem{Lem}[Theo]{Lemma}
\newtheorem{Prop}[Theo]{Proposition}
\newtheorem{Cor}[Theo]{Corollary}
\theoremstyle{definition}

\newtheorem{Ass}[Theo]{Assumption}

\definecolor{webgreen}{rgb}{0,.5,0}
\theoremstyle{rem}
\newtheorem{rem}[Theo]{Remark}

\numberwithin{equation}{section}

\allowdisplaybreaks
\begin{document}

\title{Regularity properties in a state-constrained expected utility maximization problem 
}

\author{Mourad Lazgham}

\thanks{The author acknowledges support by Deutsche Forschungsgemeinschaft through the Grant SCHI 500/3-1.}


\dedicatory{\emph{Department of Mathematics, University of Mannheim,  Germany}}

\keywords{Expected utility maximization problem, value function, price impact, optimal strategy, dynamic programming principle, Bellman's principle.}

\begin{abstract} 
We consider a stochastic optimal control problem 
in a market model with temporary and permanent price impact, which is related to an \emph{expected utility maximization} problem under finite fuel constraint. 
We establish the initial condition fulfilled by 
the corresponding value function and show its first regularity property. 
Moreover,
we can prove the existence and uniqueness of optimal strategies under rather mild model assumptions. On the one hand, this result is of independent interest. On the other hand, it will then allow us to derive further regularity properties of the corresponding value function, in particular its continuity and partial differentiability. As a consequence of the continuity of the value function, we will prove the dynamic programming principle without appealing to the classical measurable selection arguments.  \\
\end{abstract}

\maketitle

\section{Introduction}
The purpose of this paper is to investigate optimal control problems originating from a classical portfolio liquidation problem for more general 
utility functions than exponential ones. Our particular focus will be on utility functions with bounded Arrow-Pratt coefficient of absolute risk aversion. We show the existence and uniqueness of the corresponding optimal strategy, 
which is no longer deterministic in this general setting. This result then helps us to derive regularity properties of the associated value function.
\indent 
A dynamic execution strategy that minimizes expected cost was first derived in \citet{BL98}.
However, as illustrated, for instance, by the 2008 Soci\'et\'e G\'en\'erale  trading loss, we have to add to execution costs the volatility risk incurred when trading. 
This extension and the corresponding mean-variance maximization problem was treated in \citet{AC01}, in a discrete-time framework, 
where the execution costs are assumed to be linear and are split into a temporary and a permanent price impact component. 
Nevertheless, linear execution costs do not seem to be a realistic assumption in practice, as argued in \citet{A03}, and it may be reasonable to consider a nonlinear temporary impact function. 
As opposed to the temporary impact, the permanent impact has  to be linear in order to avoid quasi-arbitrage opportunities, as shown in \citet{HS04}. 
The mean-variance approach  can also be regarded as an expected-utility maximization problem for an investor with constant absolute risk aversion, which was in part solved by \citet{SST10}, where the existence and uniqueness of an optimal trading strategy, which is moreover deterministic, is proved. The latter one can be computed by solving a nonlinear Hamilton equation. Furthermore, the corresponding value function is the unique classical solution of a nonlinear degenerated  Hamilton-Jacobi-Bellman equation with singular initial condition. 
\\\indent In this paper, we generalize this framework by considering utility functions that lie between two exponential utility functions (also called CARA utility functions). 
This case was already studied for infinite-time horizons in a one-dimensional framework with linear temporary impact without drift; see \citet{SS09}, as well as \citet{S08}, where the optimal trading strategy is characterized as the unique bounded solution of a classical fully nonlinear parabolic equation. 
It was shown that the optimal liquidation strategy is Markovian and a feedback form was given. Moreover, the optimal strategy is deterministic if and only if the utility function is an exponential function. The derivation of the above results is due to the fact that,  when considering infinite time horizon, the (transformed) optimal strategy solves a \emph{classical} parabolic PDE, because the time parameter does not appear in the equation. 
In this article, we address the question of deriving the optimal liquidation strategy for the \emph{finite-time} horizon. Here we face the difficulty that commonly used change of measure techniques, involving the Dol\'eans-Dade exponential, simply go out the window. Due to this failure, 
we have to think differently and to extend our consideration to solutions that are no longer classical ones. 
\\\indent
Our  first main result 
deals with the existence and uniqueness of the optimal strategy. The proof of this result is mainly an analytical one and only requires the boundedness of the Arrow-Pratt coefficient of risk aversion of the utility function. As a direct consequence of this theorem, we can show that the associated value function is continuously differentiable in its revenues parameter (and even twice continuously differentiable if the utility function is supposed to have a convex and decreasing derivative; this condition is fulfilled if, e.g., the utility function is a convex combination of exponential utility functions). 

 After setting up our framework in Section 2.1 and making clearer our definition of utility functions with exponential growth, 
    we prove the concavity property and the initial condition fulfilled by the value function (Section 2.2). Our main results on the existence and uniqueness of the optimal strategy is given in Theorem \ref{eos}. The derivation of both results is split into several technical steps (see Section 2.3 and Section 2.4, respectively). With this at hand, we can derive the differentiability property of the value function in the revenues parameter (Theorem \ref{v_r}). The relatively involved proof of the continuity property (stated in Theorem \ref{cv}) will also follow from Theorem \ref{eos}. Using the continuity property of the value function, we conclude by establishing the underlying Bellman principle (Theorem \ref{bp}). In its proof we  face  measurability issues, and we have to restrict ourselves to considering the Wiener space to make matters clearer. This will be carried out without  referring to  measurable selection arguments, typically used in proofs of the dynamic programming principle where no a priori regularity of the value function is known to hold; see, e.g.,  \citet{M66} or  \citet{W80}, \citet{RU78}. Note that in most of the literature where the Bellman principle is related to stochastic control problems, its (rigorous) proof is simply omitted, or the reader is referred to the above literature. When the value function is supposed to be continuous, an easier version of its proof can be found in \citet{K08} or \citet{SB78}: this is however not directly applicable in our context, since we have to deal, among others, with a finite fuel constraint. 
\section{Main results}
\subsection{Modeling framework}

Let $(\Omega, \F, \P)$ be a probability space with a filtration $(\F_t)_{0\leq t\leq T}$ satisfying the usual conditions. Taking $X_0\in\R^d$, we consider a stochastic process $X_t=(X_t^1,\dots,X_t^d)$ starting in $X_0$ at time $t=0$ that has to fulfill the boundary condition $X_T=0$. For example, we can think of a basket of shares in $d$ risky assets  an investor can choose to liquidate a large market order, where we describe by $X^i_t$ the number of shares of the $i$-th asset held at time $t$. Following the notation in \citet{SSB08}, we denote by
\begin{equation}
\cR_T^{X}=R_0+\int_0^T X_t^\top\sigma\;dB_t +\int_0^T b\cdot X_t \;dt-\int_0^T f(\dot{X}_t)\;dt
\label{rp}
\end{equation}
the revenues over the time interval $[0,T]$ associated to the process $X$.  Here $R_0\in\R$, $B$ is a standard $m$-dimensional Brownian motion starting in $0$ with drift $b\in\mathbb{R}^d$ and volatility matrix $\sigma=(\sigma^{ij})\in\R^{d\times m}$, and the  nonnegative, strictly convex function $f$ has superlinear growth and satisfies the two conditions $$\lim_{|x|\longrightarrow \infty} \tfrac{f(x)}{|x|}=\infty \quad\text{ and }f(0)=0.$$
 Further, we assume that the drift vector $b$ is orthogonal to the kernel of the covariance matrix $\Sigma =\sigma\sigma ^\top$, which guarantees that there are no arbitrage opportunities for a 'small investor' whose trades do not move asset prices. The revenues processes can be interpreted economically: $R_0$  can be viewed as the face value of the portfolio (which can include a permanent price impact component), the stochastic integral models the accumulated volatility risk, whereas the second integral represents the linear drift applied to our state process. The last term stands for the cumulative cost of the temporary price impact.
Further, by
\begin{equation*}
\X_{det}(T, X_0)=\left\{X:[0,T]\rightarrow \R^d \; \text{ absolutely continuous,}\; X_0\in\R^ d\;\text{and}\; X_T=0 \right\}
\end{equation*}
we denote the set of the deterministic processes whose speed liquidation processes $\dot{X}_t$ are defined $\lambda$-a.e., where $\lambda$ is the Lebesgue-measure on $[0,T]$. Analogously, by
\begin{eqnarray*}
\lefteqn{\X(T, X_0):=}\\
&\big\{ (X_t)_{t\in[0, T]} \text{ adapted},\; t\rightarrow X_t \in \X_{det}(T, X_0), \text{ a.s.,~ and}\; \sup_{0\leq t\leq T} |X_t| \in \K (\P)\big\}
\end{eqnarray*}
we denote the set of the $\P\otimes\lambda$-a.e.~bounded stochastic processes whose speed liquidation processes $\dot{X}_t$ can be defined $\P\otimes\lambda$-a.e., due to absolute continuity.
\begin{rem}
From a hedging point of view, the absolute continuity of $X$ seems to be very restrictive, since this does not englobe the Black-Scholes Delta hedging, for example. However, from a mathematical point of view, this serves as a reasonable starting point for developing a theory of optimal control problems for functions with bounded variation. 
\xqed\diamondsuit
\end{rem}
It will be convenient to parametrize elements in $\X(T, X_0)$ as  in \citet{SSB08}.  Toward this end, for $\xi$ progressively measurable and $\xi_t$ with values in $\R^d$, for $t\leq T$, let us denote by
$$
\dot{\X}_0(T, X_0)=\Big\{ \xi \,|\, X_t=X_0-\int_0^t \xi_s\;ds \text{ a.s.~for }\; X\in \X(T,X_0) \Big\}
$$
the set of \emph{control processes} or speed processes of a given process $X$. 
From now on we will write $\cR^{\xi}$ for the revenues process associated to a given $\xi\in\dot{\X}_0(T, X_0)$, to insist on the dependence on $\xi$. The pair $(X^\xi,\cR^\xi)$  is then the solution of the following controlled stochastic differential equation:
\begin{equation}
\begin{cases}
                d\cR^{\xi}_{t}=X^{\top}_{t}\sigma dB_{t} +b\cdot X_{t}\,dt-f(-\xi_{t})\,dt,\\
                dX_t=-\xi_t\, dt,\\
                  \cR^{\xi}_{\arrowvert t=0}=R_{0}\;\text{and}\; X_{\arrowvert t=0}=X_0.
\label{sde1}\end{cases}
\end{equation}
We denote by $\dot{\X}(T, X_0)$  the subset of all control processes $\xi\in \dot{\X}_0(T, X_0)$ that satisfy the 
additional requirement 
\begin{equation}
\E\bigg[\int_{0}^{T}\big(X^\xi_t\big)^\top\Sigma X^\xi_t + |b\cdot X^\xi_t-f(\xi_t)|+|\xi_t|\,dt\bigg]<\infty.\label{ias}
\end{equation}
For convenience, we enlarge the preceding set $\dot{\X}(T, X_0)$  by introducing the notation $\dot{\X}^1(T,X_0)$ for the set of the liquidation strategies whose paths satisfy \eqref{ias}, but are not necessarily uniformly bounded:
\begin{eqnarray*}
\lefteqn{\dot{\X}^1(T,X_0)}\\
&:=&\Big\{\xi\,\big|\,\Big(X^\xi_t:=X_0-\int_0^t \xi_s\;ds\Big)_{t\in[0,T]}\text{ adapted, } t\rightarrow X^\xi_t(\omega) \in\X_{det}(T, X_0)\, \P\text{-a.s.}\Big\}\\
&\bigcap& \Big\{\xi\,\big|\,\E\bigg[\int_{0}^{T}\big(X^\xi_t\big)^\top\sigma X^\xi_t + |b\cdot X^\xi_t-f(\xi_t)|+|\xi_t|\,dt\bigg]<\infty\Big\},
\end{eqnarray*}
which is clearly a subset of $\dot{\X}(T, X_0).$
The maximization problem can thus be written in the form
\begin{equation}
\sup_{\xi\in\dot{\X}^1(T, X_0)}\E\left[u\left(\cR_{T}^{\xi}\right)\right].\label{max1}
\end{equation}

In this paper, we will consider a special class of utility functions.  These functions will have a bounded Arrow-Pratt coefficient of absolute risk aversion, i.e., we will suppose that there exist two positive constants $A_i,i=1,2,$ such that
\begin{equation}
0< A_1\leq-\frac{u''(x)}{u'(x)}\leq A_2, \quad\forall x\in\R.\label{apc}
\end{equation} 
This inequality implies that we can assume w.l.o.g. that $0<A_1<1<A_2$,
which gives us the following estimates
 \begin{equation}\exp(-A_1 x)\leq u'(x)\leq \exp(-A_2x)+1\quad \text{ for } x\in\R.\label{ieu'}\end{equation}
 and
\begin{equation}
u_1(x):=\frac1{A_1}-\exp(-A_1x)\geq u(x)\geq -\exp(-A_2x)=: u_2(x).\label{ubd1}
\end{equation}
From \cite{SST10} we know  that for exponential utility functions (that is,  utility functions of the form $a-b\exp(-c x)$, where $a\in\R$ and $b,\;c>0$) there exists a unique deterministic and continuous strategy solving the maximization problem \eqref{max1}. Moreover, the corresponding value function, i.e., the value function generated by the exponential expected-utility maximization problem, is the unique continuously differentiable solution of a Hamilton-Jacobi-Bellman equation. We will use this strong result  to establish the existence of an optimal control under the condition \eqref{ubd1}.
\noindent Here, we will study the regularity properties of the following value function:
\begin{equation}
V(T,X_0,R_0)=\sup_{\xi\in\dot{\X}^1(T, X_0)}\E\left[u\left(\cR_{T}^{\xi}\right)\right],\label{omp}
\end{equation}
where the utility function $u$ satisfies \eqref{ubd1}.
 Note that the corresponding estimates yield the following bounds for our value function
\begin{equation}
\sup_{\xi\in\dot{\X}^1(T, X_0)}\E\left[u_1\left(\cR_{T}^{\xi}\right)\right] \geq\sup_{\xi\in\dot{\X}^1(T, X_0)}\E\left[u\left(\cR_{T}^{\xi}\right)\right]\geq \sup_{\xi\in\dot{\X}^1(T, X_0)}\E\left[u_2\left(\cR_{T}^{\xi}\right)\right],
\label{max3}
\end{equation}
whence
\begin{equation}
V_1(T,X_0,R_0)=\E\left[u_1\left(\cR_{T}^{\xi^*_1}\right)\right]\geq V(T,X_0,R_0)\geq \E\left[u_2\left(\cR_{T}^{\xi_2^*}\right)\right]=V_2(T,X_0,R_0)\label{vfs},
\end{equation}
where $V_i, \;i=1,2,$ denote the corresponding exponential value functions and $\xi^*_i, \;i=1,2,$ are the corresponding optimal strategies.
\subsection{Concavity property and initial condition satisfied by the value function}
 The aim of this subsection is to prove that the  map
\[
(X,R)\longmapsto V(T,X,R)
\]
is concave, for fixed $T\in[0,\infty[$, and to derive the initial condition satisfied by $V$, where $V$ is the value function of the optimization problem as defined in \eqref{omp}. These are  fundamental properties of the value function of the considered maximization problem.\\
 \indent We start by proving the following proposition which establishes the first regularity property of the value function: the concavity of the value function in the revenues parameter, with  $T,X_0\in\;]0,\infty[\times\R^d$ being fixed. This will enable us later to prove the differentiability of the value function in the revenues parameter, other parameters being fixed, with the help of the existence of an optimal strategy. 
\begin{Prop}\label{cvf}
For fixed $T\in\;]0,\infty[$,
\[
(X,R)\longmapsto V(T,X,R)
\]
is a concave function.
\end{Prop}
\begin{proof}

Toward this end, let $X,\overline{X}\in\R^d,R,\overline{R}\in\R$ and $\lambda\in\;]0,1[$. Further, consider the strategies $\xi\in\dot{\X}^1(T,X)$ and $\overline{\xi}\in\dot{\X}^1(T,\overline{X})$. Note that $\lambda \xi + (1-\lambda) \overline{\xi}\in\dot{\X}(T,\lambda X+(1-\lambda)\overline{X})$. Let us denote 
\[
\cR^{\lambda\xi+(1-\lambda)\overline{\xi}}_T:=\int_0^T (X^{\lambda\xi+(1-\lambda)\overline{\xi}}_t)^\top\sigma\,dB_t+ \int^T_0b\cdot X^{\lambda\xi+(1-\lambda)\overline{\xi}}_t\,dt-\int^T_0f(-\lambda\xi+(1-\lambda)\overline{\xi}_t)\,dt.\\
\]
We then have for fixed $\xi,\overline{\xi}$:
\begin{align*}
\lefteqn{
V(T,\lambda X+(1-\lambda)\overline{X},\lambda R+(1-\lambda)\overline{R}))}&\\
 &\geq\E\big[u\big(\lambda R+ (1-\lambda) \overline{R}+ \cR^{\lambda\xi+(1-\lambda)\overline{\xi}}_T\big)\big]\\
&\geq \E\big[u\big(\lambda R+ (1-\lambda) \overline{R})+ \lambda \cR^{\xi}_T+(1-\lambda) \cR^{\overline{\xi}}_T\big)\big]\\
&\geq \lambda  \E\big[u\big( R+\cR^{\xi}_T\big)\big]+(1-\lambda) \E\big[u\big(\overline{R}+ \cR^{\overline{\xi}}_T\big)\big],
\end{align*}
where the first inequality is due to the definition of the value function $V$ at\\ $(\lambda X+(1-\lambda)\overline{X},\lambda R+(1-\lambda)\overline{R})$, and the second one follows from the fact that $\xi \mapsto \cR^\xi_T$ is concave and $u$ is increasing. Finally, the third one is due the concavity of $u$. Taking now the supremum over $\xi$ ($\overline{\xi}$ being fixed), we obtain
\begin{equation*}
V(T,\lambda X+(1-\lambda)\overline{X},\lambda R+(1-\lambda)\overline{R}))\geq \lambda  V(T,X,R)+(1-\lambda) \E\big[u\big(\overline{R}+ \cR^{\overline{\xi}}_T\big)\big].
\end{equation*}
Taking the supremum over $\overline{\xi}$ in the preceding equation, we  obtain
\[
V(T,\lambda X+(1-\lambda)\overline{X},\lambda R+(1-\lambda)\overline{R}))\geq \lambda V(T,X,R)+ (1-\lambda)V(T,X,\overline{R}),
\]
which yields the assertion.
 \end{proof}
Further, we establish the initial condition fulfilled by the value function. 
\begin{Prop}\label{icv}
Let $V$ be the value function of the maximization problem \eqref{omp}. Then $V$ fulfills the following initial condition
\begin{align}
V(0,X,R)= \lim_{T\downarrow 0}V(T,X,R)&=\begin{cases}
                                             u(R),& \text{if}\; X=0,\\
                                     -\infty,& \text{otherwise}.\\
                                    \end{cases}\label{hjbic}
           \end{align}
\end{Prop}
\begin{proof}

We first note that if $X\neq 0$, then 
\[
\lim_{T\rightarrow 0} V(T,X,R)=-\infty,
\]
because $V$ is supposed to lie between two CARA value functions which tend to $-\infty$ as $T$ goes to zero, if $X\neq 0$ (see \citet{SST10}). Suppose now that $X=0$. We want to show that
\[
\lim_{T\rightarrow 0} V(T,0,R)=u(R).
\]
Observe first that
\[
V(T,0,R)\geq \E\Big[u\Big(\cR^\xi_T\Big)\Big]=u(R),
\]
by choosing the strategy $\xi_t=0\;\text{for all}\; t\in[0,T],\;T>0$.  Since $V$ is increasing in $T$, for fixed $X,R$, the limit $\lim_{T\rightarrow 0} V(T,X,R)$ exists, which implies that
\[
\lim_{T\rightarrow 0} V(T,0,R)\geq u(R).
\]
We now prove the reverse inequality
\begin{equation}
\lim_{T\rightarrow 0} V(T,0,R)\leq u(R).\label{u(r)}
\end{equation}
Let $\xi$ be a round trip starting from $0$ (i.e: $\xi\in\dot{\X}^1(T,0)$). Applying Jensen's inequality to the concave utility function $u$, we get
\begin{equation*}
\E\big [u\big(\cR^\xi_T\big)\big]\leq u\bigg(R+ \E\bigg[ \int_0^T b\cdot X^\xi_t \;dt-\int_0^T f(-\xi_t)\;dt\bigg]\bigg).
\end{equation*}
 We have to show now 
\begin{equation}
\limsup_{T\downarrow 0}\E\bigg[ \int_0^T b\cdot X^\xi_t \;dt-\int_0^T f(-\xi_t)\;dt\bigg]\leq 0.\label{lif}
\end{equation}
To this end we use the integration by parts formula to infer
\[
\int_0^T b\cdot X^\xi_t \;dt=\int_0^T tb\cdot\xi_t\,dt.
\]
Hence, we have
\begin{align*}
\E\bigg[ \int_0^T b\cdot X^\xi_t \;dt-\int_0^T f(-\xi_t)\;dt\bigg]&=\E\bigg[ \int_0^T tb\cdot\xi_t- f(-\xi_t)\;dt\bigg]\\
&\leq \int_0^T f^*(-bt)\,dt,\\
\end{align*} 
where $f^*$ designates the Fenchel-legendre transformation of the convex function $f$. Note that $f^*$ is a finite convex function, due to the assumptions on $f$ (see Theorem 12.2 in \cite{R97}), and in particular continuous, so that
\[
\int_0^T f^*(-bt)\,dt\underset{T\downarrow 0}{\longrightarrow} 0,
\]
which proves \eqref{lif}. Finally, using that $u$ is continuous and nondecreasing, we get
\begin{align*}
\lim_{T\rightarrow 0} V(T,0,R) &\leq\liminf_{T\rightarrow 0}   \sup_{\xi\in\dot{\X}^1(T,0)}u\bigg(R+ \E\bigg[ \int_0^T b\cdot X^\xi_t \;dt-\int_0^T f(-\xi_t)\;dt\bigg]\bigg)\\
 &\leq u(R).
\end{align*}
\end{proof}
\subsection{Existence and uniqueness of an optimal strategy}
In this section we aim at investigating the existence and uniqueness of an optimal strategy for the maximization problem
\begin{equation*}
\sup_{\xi\in\dot{\X}^1(T, X_0)}\E[u(\cR_{T}^{\xi})],
\end{equation*}
where $u$ is strictly concave, increasing and satisfies \eqref{ubd1}. The quantity $\cR^\xi_T$ denotes the revenues associated with the liquidation strategy $\xi$ over the time interval $[0, T]$. 
The next theorem establishes the main result of the current section.
\begin{Theo}\label{eos}
Let  $\left(T,X_0,R_0\right)\in\;]0,\infty[\times\R^d\times\R$, then there exists a unique optimal  strategy  $\xi^*\in\dot{\X}^1(T,X_0)$ for the maximization problem \eqref{omp}, which satisfies  
\begin{equation}
V(T,X_0,R_0)=\sup_{\xi\in\dot{\X}^1(T, X_0)}\E[u(\cR_{T}^{\xi})]=\E\Big[u\big(\cR^{\xi^*}_T\big)\Big].\label{omp1}
\end{equation}
\end{Theo}
\noindent 
The main idea of the  proof is to show that a sequence of strategies $\left(\xi^n\right)$ such that the corresponding expected utilities converge from below to the supremum, i.e., 
$$
\;\E\big[u\big(\cR_{T}^{\xi^{n}}\big)\big]\nearrow\sup_{\xi\in\dot{\X}^1(T, X_0)}\E\big[u\big(\cR_{T}^{\xi}\big)\big],
$$
 lies in a weakly sequentially compact subset of $\dot{\X}^1(T, X_0)$, due to the fact that the function $u$ satisfies the inequalities \eqref{ubd1}. 
Then we can choose a subsequence 
that converges weakly to the strategy $\xi^*$. The uniqueness of the optimal strategy will follow from the strict concavity of the map $\xi\longmapsto \E[u(\cR_{T}^{\xi})]$.
\begin{rem}
Note that due to inequality \eqref{vfs}, we can w.l.o.g suppose that the above sequence verifies
\begin{equation}
\E\left[\exp(-A_1\cR_{T}^{\xi^{n}})\right]\leq 1+1/A_1-V_2(T,X_0,R_0),\quad\text{for all }n\in\N,\label{enf}
\end{equation} 
where $V_2$ denotes the following CARA value function: 
\[
V_2(T,X_0,R_0)=\sup_{\xi\in\dot{\X}^1(T, X_0)}\E\big[-\exp\big(-A_2\cR_{T}^{\xi}\big)\big].
\]
\end{rem}
We will split the proof into several steps. First, we will prove a weak compactness property of certain subsets of $ \dot{\X}^1(T,X_0)$. Let us start by recalling some fundamental functional analysis results. The first one is a classical characterization of convex closed sets (see, e.g., \citet{SF11}, Theorem A.60).
\begin{Theo}\label{wccf}
Suppose that $E$ is a locally convex space and that $\cC$ is a convex subset of $E$. Then $\cC$ is weakly closed if and only if $\cC$ is closed with respect to the original topology of $E$.
\end{Theo}
\begin{Cor}\label{clsc}
Let $\varphi:E\rightarrow]-\infty;\infty]$ be a lower semi-continuous convex function with respect to the original topology of $E$. Then $\varphi$ is lower semi-continuous with respect to the weak topology $\sigma(E',E),$ where $E'$ denotes the dual space of $E$.
In particular, if $(x_n)$ converges weakly to $x$, then
\begin{equation}
\varphi(x)\leq\liminf\varphi(x_n).\label{sl}
\end{equation}
\end{Cor}
\begin{proof}
See, e.g., \citet{B11}. 
\end{proof}
 \begin{Cor}\label{lscc}
Let $(S,\S, \mu)$ be a measurable space, $F:\R^d\rightarrow\R$ a convex function bounded from below, and $(x_n)\subset \J((S,\S,\mu);R^d)$.  Suppose that $(x_n)$ converges to $x$, weakly. Then
 \[
 \int F(x)d\mu\leq\liminf\int F(x_n)d\mu.
\]
Further, if we suppose that $F:\R^d\rightarrow\R$ is concave and bounded from above, we have an analogous conclusion, i.e., 
 \[
 \int F(x)d\mu\geq\limsup\int F(x_n)d\mu.
\]
\end{Cor}
\begin{proof}
We only show the first assertion. 
 Using the preceding corollary, it is sufficient to prove that 
the convex map
\begin{align*}
L^1((S,\S,\mu);\R^d)& \longrightarrow [0,\infty]\\
\alpha&\longmapsto \int F(\alpha)\,d\mu
\end{align*}
 is lower semi-continuous  with respect to the strong topology of $\J((S,\S,\mu);\R^d)$. To this end, let  $c\in\R$ and $(x_n)\subset \J((S,\S,\mu);\R^d)$ be a sequence that converges strongly to some $x\in \J((S,\S,\mu);\R^d)$ and satisfies the condition $ \int F(x_n)\,d\mu\leq c$. We have to show that
\[
\int F(x)\,d\mu\leq c.
\]
Taking a subsequence, if necessary, we can suppose that $(x_n)$ converges to $x$ $\mu$-a.e. Applying then Fatou's Lemma, we infer 
\[
\int F(x)\,d\mu=\int \liminf F(x_n)\,d\mu\leq \liminf \int F(x_n)\,d\mu\leq c,
\]
which concludes the proof.
\end{proof}
 With this at hand, we can show the following lemma, which will be useful for us to prove the continuity of the value function.
\begin{Lem}\label{wcx}
Let $(X^n_0, T^n)\subset\R^d\times \R$ be a sequence that converges to $(X_0,T)$ and set $\overline{T}:=\sup_n T^n$. Moreover, consider a sequence $(\zeta^n)$  in $\dot{\X}^1(T^n,X^n_0)$ and take a constant $c>0$ such that
\begin{equation}
\E\bigg[\int^T_0f(-\zeta^n_t)\,dt\bigg]\leq c.
\end{equation}
Suppose that $\left(\zeta^n\right)$ converges to $\zeta$ with respect to the weak topology in $$\J:=\L ^1\Big(\Omega\times [0,T], \F\otimes \B([0,T]), (\P\otimes\lambda)\Big).$$ Then $\zeta\in\dot{\X}^1(T,X_0)$ and 
\begin{equation}
\E\bigg[\int^T_0f(-\zeta_t)\,dt\bigg]\leq c.
\end{equation}
\end{Lem}

\begin{proof}
First note that we have the canonical inclusion $\dot{\X}^1(T^n,X^n_0)\subseteq\dot{\X}^1(\overline{T},X^n_0)$, by setting $\zeta^n=0$ on $[T^n,\overline{T}]$. Now, we wish to prove that $\int^T_0\zeta\,dt= X_0$. Suppose by  way of contradiction that $\int^T_0\zeta\,dt\neq X_0.$ Then, there exists a component $\zeta^i$ 
such that $\int_0^T\zeta^i_t\,dt\neq X^i_0.$ Thus, we can assume without loss of generality that $d=1$ 
and work toward a contradiction. Under this assumption, there exists a measurable set $\A$ with $\P(\A)>0$, such that 
$
\int^T_0\zeta_t\d> X_0 \text{ on } \A, \text{ or } \int^T_0\zeta_t\d< X_0 \text{ on } \A.
$
Without loss of generality, we can assume that  
\begin{equation}
\int^T_0\zeta_t\d> X_0 \quad\text{on } \A.\label{iA}
\end{equation}
Because  $\zeta^n\in\dot{\X}^1(T^n,X^n_0)$ converges to $\zeta$, weakly in $\J$, we have
\begin{align*}
 0=\E\left[\bigg(X^n_0-\int^{T^n}_0 \zeta^n_t\d \bigg)\b1_{\A}\right]&=\E\left[\bigg(X^n_0-\int^{\overline{T}}_0 \zeta^n_t\d \bigg)\b1_{\A}\right]\\
 &\longrightarrow \E\left[\bigg(X_0-\int^{\overline{T}}_0 \zeta_t\d\bigg)\b1_{\A}\right]=0.
\end{align*}
If $\overline{T}=T$ the result is proved, because the expectation on the right-hand side has to be negative, due to the assumption \eqref{iA}; this is a contradiction.

Suppose now that $\overline{T}>T$. It is sufficient  to show that $\zeta=0$ on $[T,\overline{T}]$. To this end, set 
$$\eta_t(\om):=\b1_{\{\zeta_t(\om)>0\}}\b1_{[T,\overline{T}]}(t).$$ Analogously, we get \[
0=\E\left[\int^{\overline{T}}_{T^n} \zeta^n_t \eta_t\d\right] \longrightarrow \E\left[\int^{\overline{T}}_T \zeta_t \eta_t\d\right]=0,
\] 
due to the weak convergence of $\zeta^n$ to $\zeta$, the fact that  
$\eta\in \K \big(\big(\Omega\times [0,\overline{T}], \F\otimes \B([0,\overline{T}]), (\P\otimes\lambda);\R^d\big)\big),$ and $\zeta_n=0$ on $[T^n,\overline{T}].$
Thus,  $\{\zeta_t(\om)>0; t\in [T,\overline{T}]\}$ is a null set.
Taking
$\eta_t(\om):=\b1_{\{\zeta_t(\om)>0\}}\b1_{[T,\overline{T}]}(t)$, we can prove in the same manner that $\{\zeta_t(\om)<0 \text{ on }[T,\overline{T}]\}$ is a null set. Hence, $\zeta=0$ on $[T,\overline{T}]$ and  therefore
$\int^T_0\zeta\,dt= X_0.$

\noindent Using Corollary \ref{lscc} we infer
\[
\E\bigg[\int^T_0f(-\zeta_t)\,dt\bigg]\leq  \liminf_{n\longrightarrow \infty}\E\bigg[\int^T_0f(-\zeta^n_t)\,dt\bigg] \leq c,
\]
\ignore{It is left to show that $\zeta$ fulfills \eqref{ias}. Since $(\zeta^n)\subset\J$ and converges weakly to $\zeta$, we also have  $\zeta\in\J$, whence
\[
\E\bigg[\int_{0}^{T}\big(X^{\textcolor{red}{\zeta?}}_t\big)^\top\sigma X^{\textcolor{red}{\zeta?}}_t + |{\textcolor{red}{\zeta?}}_t|\,dt\bigg]\leq \E\bigg[\int_{0}^{T}\|\zeta\|^2_{\J} |\Sigma|+ |{\textcolor{red}{\zeta?}}_t|\,dt\bigg]<\infty.
\]
Further, we have
\[
\E\bigg[\int_{0}^{T} |b\cdot X^\zeta_t-f(\zeta_t)\,dt\bigg]\leq\E\bigg[\int_{0}^{T} |b|\|\zeta\|_{\J}+f(\zeta_t)dt\bigg]\leq |b|\|\zeta\|_{\J}T +c <\infty,
\]
which shows that \eqref{ias} is verified. Therefore $\zeta\in\dot{\X}^1(T,X_0)$.}
which concludes the proof.
\end{proof}
We can now  prove a weak compactness property of a certain family of subsets of $\dot{\X}^1(T,X_0)$.
\begin{Prop}\label{wsc}
For $c>0$, let 
\[
\overline{K}_{c}:=\Big\{\xi\in \dot{\X}^1(T,X_0)\big|\; \E\bigg[\int^T_0f(-\xi_t)\,dt\bigg]\leq c\Big\}.
\]
Then $\overline{K}_{c}$ is a weakly sequentially compact subset of 
$$\J:=\L ^1\big(\big(\Omega\times [0,T], \F\otimes \B([0,T]), (\P\otimes\lambda)\big);\R^d\big).$$
\end{Prop}
\begin{proof}
We first prove that $\overline{K}_{c}$ is a closed convex set with respect to the strong topology of 
 $\J$.
 
 The convexity of $\overline{K}_{c}$ is a direct consequence of the convexity of the map
$$
\xi\longmapsto \E\bigg[\int^T_0 f(-\xi_t)\d\bigg].
$$
  To show that $\overline{K}_{c}$ is closed, let  $\xi^n$ be a sequence in $\overline{K}_{c}$ that converges strongly to $\xi$. Then, in particular, $\xi^n$ converges to $\xi$ weakly and we are in the setting of Lemma 
 \ref{wcx}, which proves that $\xi\in\overline{K}_{c}$. Thus, $\overline{K}_{c}$ is convex and closed in $\L^1$. Hence, it is  also closed with respect to the weak topology, as argued in Theorem \ref{wccf}. To prove that $\overline{K}_{c}$ is weakly sequentially compact, it remains to show that $\overline{K}_{c}$ is uniformly integrable, by the Dunford-Pettis theorem (\cite{DS88}, Corollary IV.8.11).

To this end, take $\e>0$ and $\xi\in\overline{K}_{c}$. There exists  a constant $\alpha>0$ such that $\frac{|\xi_t|}{f(-\xi_t)}\leq \frac \e{c}$ for $\big|\xi_t\big|>\alpha$, due to the superlinear growth property of $f$. Because $f(x)=0$ if and only if $x=0$, the quantity $1/f(-\xi_t)$ is well-defined on $\{|\xi_t|>\alpha\}$ and we obtain
\begin{align*}
\E\bigg[\int^T_0\b1_{\{|\xi_t|>\alpha\}}\big|\xi_t\big|\,dt\bigg]
&\leq \E\bigg[\int^T_0\b1_{\{|\xi_t|>\alpha\}}f(-\xi_t)\,dt\bigg] \frac \e{c}\leq\e,
\end{align*}
which proves the uniform integrability of $\overline{K}_{c}$. 
\end{proof}
In the next lemma, we give a lower and an upper bound for the non-stochastic integral terms that appear in the revenue process.
\begin{Lem}\label{bin}
Suppose that $b\neq0,$ and let $\xi\in\dot{\X}^1(T,X_0)$ and $t^1,t^2\in [0,T].$  Then there exists a constant $C>0$, depending on $f,b$ and $T$, such that 
\begin{IEEEeqnarray*}{rCl}
\lefteqn{ -\frac54\int^{t_2}_{t_1} f(-\xi_t)\,dt  -|b|CT^2/2-b\cdot\big(t_1X^\xi_{t_1}-t_2X^\xi_{t_2}\big)}\\
&\leq&\int^{t_2}_{t_1}\left(b\cdot X^\xi_t-f(-\xi_t)\right)\,dt\leq -\frac34\int^{t_2}_{t_1} f(-\xi_t)\,dt  + |b|CT^2/2-b\cdot\big(t_1X^\xi_{t_1}-t_2X^\xi_{t_2}\big).
\end{IEEEeqnarray*}
\end{Lem}
\begin{proof}
Set $\gamma:=\frac1{4|b|T}$. Because $\lim_{|x|\longrightarrow \infty}\frac{|x|}{f(x)}=0$, there exists a constant $C_\gamma=C>0$ such that $\tfrac {|y|}{f(y)}\leq \gamma$ for $|y|> C$. Consider now the set $A_t:=\{|\xi_t|\leq C\}$. Then we have using integration by parts: 
\begin{IEEEeqnarray*}{rCl}
  \lefteqn{\int^{t_2}_{t_1}\left(-b\cdot X^\xi_t+f(-\xi_t)\right)\,dt}\\
&\geq& b\cdot\big(t_1X^\xi_{t_1}-t_2X^\xi_{t_2}\big) -\int^{t_2}_{t_1} \b1_{A_t}|b\cdot\xi_t| t\,dt + \int^{t_2}_{t_1} \b1_{A_t}f(-\xi_t)\,dt\\
&&+\> \int^{t_2}_{t_1} \b1_{A_t^c}f(-\xi_t)\Big(1+ \frac{b\cdot\xi_t t}{f(-\xi_t)}\Big)\,dt\\
&\geq& b\cdot\big(t_1X^\xi_{t_1}-t_2X^\xi_{t_2}\big)+\frac14 \int^{t_2}_{t_1} \b1_{A_t}f(-\xi_t)\,dt+\frac34 \int^{t_2}_{t_1} f(-\xi_t)\,dt  - |b|CT^2/2,
\end{IEEEeqnarray*}
using the above estimates. This proves the lower inequality. To prove the upper inequality, it is sufficient to follow step by step the preceding arguments and to give an upper bound of the corresponding terms, instead of a lower bound.
\end{proof}

The subsequent lemma shows that a sequence of strategies in $\dot{\X}^1(T,X_0)$ such that 
the corresponding expected utilities converge to the supremum in \eqref{omp1}  can be chosen in a way that it belongs to some $\overline{K}_{m}$, for $m$ large enough. This will be crucial  for proving the existence of an optimal strategy. Here, we will use  the fundamental property \eqref{enf} satisfied by the  sequence $(\xi^n)$. 
\begin{Lem}\label{osp} Let $(\xi^n)$ be a sequence of strategies such that 
\begin{equation} 
 \xi^n\in \dot{\X}^1(T, X_0)\; \text{and } \;\E\left[u\left(\cR_{T}^{\xi^{n}}\right)\right]\nearrow\sup_{\xi\in\dot{\X}^1(T, X_0)}\E\left[u\left(\cR_{T}^{\xi}\right)\right].\label{sxg}
\end{equation}
Then there exists a constant $m>0$ such that
 $$\xi^n\in \overline{K}_{m}=\Big\{\xi\in \dot{\X}^1(T,X_0)\big|\; \E\bigg[\int^T_0f(-\xi_t)\,dt\bigg]\leq m\Big\},$$ for every $n\in\N$.
\end{Lem}
\begin{proof}
Set $\overline{M}:=\overline{M}(T,X_0,R_0)=1+1/A_1-V_2(T,X_0,R_0)$. We first note that, due to \eqref{enf}, we have 
\[
 \E\bigg[e^{-A_1\Big(R_0+\int_0^T (X^{\xi^n}_t)^\top\sigma\,dB_t+\int^T_0b\cdot X^{\xi^n}_t\,dt-\int^T_0f(-{\xi^n}_t)\,dt\Big)}\bigg]\leq 1/A_1-V_2(T,X_0,R_0)= \overline{M}.
\]
We want to show that 
\begin{equation}
\xi^n\in \widetilde{K}_{\alpha}:=\bigg\{\xi\in \dot{\X}^1(T,X_0)\big|\; \E\bigg[\int^T_0-b\cdot X^\xi_t+f(-\xi_t)\,dt\bigg]\leq \alpha\bigg\},\label{kca}
\end{equation}
for $\alpha\geq\frac{\overline{M}-1}{A_1}+R_0$. To prove \eqref{kca}, we use the fact that $e^x\geq 1+x$, for all $x\in\R,$ as well as the martingale property of $Y_T:=\int_0^T (X^{\xi^n}_t)^\top\sigma\,dB_t$ (which is satisfied, due to \eqref{ias}), whence we infer
\begin{align*}
 \overline{M}& \geq
\E\bigg[-A_1\Big(R_0+\int^T_0b\cdot X^{\xi^n}_t\,dt-\int^T_0f(-\xi^n_t)\,dt\Big)\bigg]+1.
\end{align*}
Then
$$
\E\bigg[\int^T_0-b\cdot X^{\xi^n}_t+f(-\xi^n_t)\,dt\bigg]\leq \frac{\overline{M}-1}{A_1}+R_0,
$$
and therefore \eqref{kca} is true.

 Using now Lemma \ref{bin}  we obtain (when setting $N:=|b|CT^2$):
\begin{equation*}
\alpha\geq\frac{\overline{M}-1}{A_1}+R_0\geq E\bigg[\int^T_0-b\cdot X^{\xi^n}_t+f(-\xi^n_t)\,dt\bigg]\geq \frac34\E\bigg[\int^T_0 f(-\xi^n_t)\,dt\bigg] - N.
\end{equation*}
Finally, for $m\geq \frac43 (\alpha+N)$  we get
\[
\E\bigg[\int^T_0 f(-\xi^n_t)\,dt\bigg]\leq m,
\]
which shows that $\xi^n\in \overline{K}_{m}$. 
\end{proof}
\begin{rem}\label{km}
Due to the preceding lemma, we can w.l.o.g assume that the supremum  in \eqref{omp1}  can be taken over strategies that belong to the set  $ \overline{K}_{m}$, for suitable $m$. More precisely, 
 \eqref{omp1} becomes 
\begin{equation}
V(T,X_0,R_0)=\sup_{\xi\in\dot{\X}^1(T, X_0)}\E\left[u\left(\cR_{T}^{\xi}\right)\right]=\sup_{\xi\in \overline{K}_{m}}\E\left[u\left(\cR_{T}^{\xi}\right)\right],
\end{equation}
where $m$ has to be chosen such that
\begin{equation}
m\geq \frac43\Big(\frac{-V_2(T,X_0,R_0)}{A_1}+R_0 +N\Big)\label{m&M}.
\end{equation}

\xqed{\diamondsuit}
\end{rem}
In the following, we will prove a fundamental property of the map $\xi\longmapsto \E\Big[u\big(\cR^\xi_T\big)\Big]$, which we will also use to prove the continuity of the 
value function for  the underlying  maximization problem.
\begin{Prop}\label{wtc}
The map 
$\xi\longmapsto \E\Big[u\big(\cR^\xi_T\big)\Big]$ is upper semi-continuous on $\dot{\X}^1(T, X_0)$ with respect to the weak topology in $\J$.
\end{Prop}
\begin{proof}
Since the map $\xi\longmapsto \E\Big[u\big(\cR^\xi_T\big)\Big]$ is concave, it is sufficient to show that the preceding map  is upper semi-continuous with respect to the strong topology of $\J$, due to Corollary \ref{clsc}. Toward this end, let   $(\widetilde{\xi}^n)$ be a sequence in $\dot{\X}^1(T,X_0)$  that converges to $\xi\in\dot{\X}^1(T,X_0)$, strongly in $\J$. Since we are dealing with a metric space, we can use the following characterization of upper semi-continuity at $\xi$:
\begin{equation}
\limsup_{k}\E\Big[u\big(\cR^{\widetilde{\xi}^{n_k}}_T\big)\Big] \leq  \E\Big[u\big(\cR^\xi_T\big)\Big].\label{ls}
\end{equation}
But we also have that  $\widetilde{\xi}^n$ converges  weakly to $\xi$ and hence we can directly apply Corollary \ref{lscc} to obtain \eqref{ls}.
\end{proof}
Now we are ready for the proof of the existence and uniqueness of the optimal strategy. 

\begin{proof}[Proof of Theorem \ref{eos}]
Let  $(\xi^{n})_{n\in\N}$ be such that 
\begin{equation*} 
 \xi^n\in \dot{\X}^1(T, X_0,R_0)\; \text{ and } \;\E\left[u\left(\cR_{T}^{\xi^{n}}\right)\right]\nearrow\sup_{\xi\in\dot{\X}^1(T, X_0)}\E\left[u\left(\cR_{T}^{\xi}\right)\right].
\end{equation*} 
 Lemma \ref{osp} implies that there exists a subsequence $\left(\xi^{n_k}\right)$ of $\left(\xi^n\right)$ and some $\xi^*\in\dot{\X}^1(T,X_0)$ such that $\xi^{n_k}\longrightarrow\xi^*$, weakly in $\J$. Due to Proposition \ref{wtc}, we get
\begin{align*}
V(T,X_0,R_0)&=\limsup_{k}  \E\Big[u\big(\cR^{\xi^{n_k}}_T\big)\Big]\leq  \E\Big[u\big(\cR^{\xi^*}_T\big)\Big],
\end{align*}
which  proves that $\xi^*$ is an optimal strategy for the maximization problem \eqref{omp}.
The uniqueness of the optimal strategy is a direct consequence of  the convexity  of $\dot{\X}^1(T, X_0)$ and (strict) concavity of $\xi\longmapsto \E[u(\cR_{T}^{\xi})]$.
\end{proof}
It is established in \citet{SST10} that the optimal strategies for CARA value functions are such that the corresponding revenues  have finite exponential moments, i.e.,
$
\E\left[\exp\big(-\lambda\cR^{\xi^{*,i}}_T\big)\right]<\infty,
$
for all $\lambda>0$, where $\xi^{*,i}$ are the optimal strategies for the value functions with respective CARA coefficients $A_1$ and $A_2$. This is due to the fact that the optimal strategies are deterministic, and hence $\int_0^T (X^{\xi^{*,i}}_t)^\top\sigma\,dB_t$ have finite exponential moments. However, for the optimal strategy in \eqref{omp1}, we only have $
\E\Big[\exp\big(-\lambda\cR^{\xi^{*}}_T\big)\Big]<\infty
$ if $\lambda\leq A_1$. 
 But otherwise (for $\lambda>A_1$) it is not clear whether or not the analogue holds. Thus, in order to avoid integrability issues, we will have to make the following assumptions. 
\begin{Ass}\label{ar}
We suppose  that the moment generating function of the revenues of the optimal strategy, denoted by $M_{\cR_{T}^{\xi^*}},$ is defined for $2A_2$, where we set 
\[
M_{\cR_{T}^{\xi^*}}(A):=\E\big[\exp(-A\cR_{T}^{\xi	^*}\big)\big].
\]
Thus, we will restrict ourselves  to  the following set of strategies:
\begin{equation}
\dot{\X}^1_{2A_2}(T, X_0):=\Big\{ \xi \in\dot{\X}^1(T,X_0)\,|\, \E\big[\exp(-2A_2\cR_{T}^{\xi}\big)\big]\leq M_{\cR_{T}^{\xi^*}}(2A_2)+1\Big\}\label{4a2}.
\end{equation}
 \end{Ass}
 \begin{Prop}
 The set $\dot{\X}^1_{2A_2}(T, X_0)$ 
is a closed convex set with respect to the strong topology in $\J$ (and hence with respect to the weak topology).
 \end{Prop}
 \begin{proof} Due to the convexity of the map $\xi\mapsto \E[\exp(-A(\cR_{T}^{\xi})]$, the preceding set is convex. To show that it is closed in $\J$, we take a sequence $(\zeta^n)$ in $\dot{\X}^1_{2A_2}(T, X_0,R_0)$ that converges to $\zeta$ in $\J$. Since $\zeta^n$ in particular converges  weakly to $\zeta$, we can use Corollary \ref{lscc} to obtain
\[
\E\big[\exp(-2A_2\cR_{T}^{\zeta}\big)\big]\leq\liminf\E\big[\exp(-2A_2\cR_{T}^{\zeta^n}\big)\big]\leq M_{\cR_{T}^{\xi^*}}(2A_2)+1,
\]
which completes the proof.
\end{proof}

 \begin{rem}\label{??}
 As argued before, if $M_{\cR_{T}^{\xi^*}}(2A2)<\infty$, then we also have $$M_{\cR_{T}^{\xi^*}}(A)<\infty\quad \text{ for all }\quad 0<A<2A_2.$$ Note that if we suppose that $u$ is a convex combination of CARA utility functions, then $M_{\cR_{T}^{\xi^*}}$ is defined on $[A_1, A_2]$.  However, we need $M_{\cR_{T}^{\xi^*}}(2A_2)$ to be  well-defined, since we will have to apply the Cauchy-Schwarz inequality to prove the continuity of the value function. 
 \end{rem}

\section{Regularity properties of the value function and the dynamic programming principle}
\subsection{Partial Differentiability of the value function}
In this section, we  will establish that the value function $V$ is continuously differentiable with respect to the parameter $R\in\R$, for fixed $(T,X)\in\;]0,\infty[\times\R^d$. Surprisingly, we just need the existence and uniqueness of the optimal strategy to prove it.  Compared to the proof of the continuity of the value function in its parameters, this one is
essentially easier, due to fact that, for fixed $T, X_0$, the value function is concave as showed in Proposition \ref{cvf}.

Further, we need to  prove the following result.

\begin{Prop}\label{u'}
Let $\xi\in \dot{\X}^1_{2A_2}(T,X_0)$. Then, the map
$
R_0\longmapsto \E\big[u\big(\cR_{T}^{\xi}+R_0\big)\big]
$
is twice differentiable on $\R$ with first and second derivative given by $\E\big[u'\big(\cR_{T}^{\xi}\big)\big]$ and $\E\big[u''\big(\cR_{T}^{\xi}\big)\big]$, respectively. 
\end{Prop}
Before beginning with the proof, we  need to prove the following lemma.
\begin{Lem}\label{ci}
Let $g$ be a real-valued locally integrable function on 
$[0,\infty[$ such that
\begin{equation}
\int^x_0 g(t)\,dt\geq 0, \quad \text{ for all }x>0.\label{ii}
\end{equation}
Then
$
\limsup_{x\rightarrow \infty} g(x)\geq 0.
$
\end{Lem}
\begin{proof}
Suppose that there exists $\e>0$ such that
$
\limsup_{x\rightarrow \infty} g(x)< -2\e.
$
Then there exists $x_0>0$ such that $g(x)\leq -\e\text{ for all } x\geq x_0,$ whence we get
\begin{align*}
\int^x_0g(t)\,dt
\leq \int^{x_0}_0 g(t)\,dt-\e(x-x_0)<0 \quad\text{ for } x \text{ large enough},
\end{align*}
which is in contradiction with \eqref{ii}.
\end{proof}
\begin{proof}[Proof of Proposition 3.1]
By translating $u$ horizontally if necessary, we can assume without loss of generality that $R_0=0$. Thus, we have to prove that the map
$
r\mapsto \E\big[u\big(\cR_{T}^{\xi}+r\big)\big]
$
is differentiable at $r=0$ with derivative $\E\big[u'\big(\cR_{T}^{\xi}\big)\big]$.
Since $u$ is concave, increasing, and lies in $C^1(\R)$, $u'$ is decreasing and positive,  hence it is sufficient to prove 
\begin{equation}
\E\big[u'\big(\cR_T^\xi-1\big)\big]<\infty.\label{u'i}
\end{equation}
Due to inequalities \eqref{ubd1}, 
we get
\[
\exp(A_2x)+u(-x)=\int_0^x \Big(\frac1{A_2}\exp(A_2x)-u'(-x)\Big)\,dx +u(0)-\frac1{A_2}\geq 0, \; x\geq 0.
\]
Hence, by translating $u$ vertically if necessary, the conditions of Lemma \ref{ci} apply with $g(x)=\frac1{A_2}\exp(A_2 x)-u'(-x)\text{ on } [0,\infty[$. Therefore, we can find a constant $C>0$ such that
\[
u'(-x)\leq C(\exp(A_2 x)+1)\quad \text{ for all } x\geq 0.
\]
Thus,
\begin{align*}
\E\big[u'\big(\cR_T^\xi-1\big)\big]
&\leq C(\E\big[\exp\big(-A_2 \cR_T^\xi\big)\big]+1) +\E\big[u'\big(\cR_T^\xi-1\big)\b1_{\{\cR_T^\xi-1\geq0\}}\big]<\infty,
\end{align*}
since $u'$ is bounded on $[0,\infty[$ and $\E\big[\exp\big(-A_2 \cR_T^\xi\big)\big]<\infty$, due to the assumption on $\xi$.  This shows the assertion for the first derivative. For the second one, we take $0<\eta<1$ and $r\in\;]-\eta,\eta[$. We wish to prove that
\begin{equation}
\sup_{r\in\;]-\eta,\eta[}\E\big[\big|u''\big(\cR_T^\xi+r\big)\big|\big]<\infty.
\end{equation}
To this end, we use inequality \eqref{apc} to obtain 
\begin{align*}
\E\big[\big|u''\big(\cR_T^\xi+r\big)\big|\big]
 &\leq\E\big[A_2u'\big(\cR_T^\xi-1\big)\big]<\infty,
 \end{align*}
which completes the proof.  \end{proof} 
In our case,  the optimal strategy depends on the parameter $R$ \emph{without}, a priori, any known control of this dependence.
Since the concavity property of the value function will be the key to establishing the desired regularity properties, we 
 consider now a family of concave $C^1$-functions $f_\alpha:\R\longrightarrow \R$ and define
$$f(x)=\sup_{\alpha} f_\alpha(x).$$
 Note that the supremum is not necessarily concave. However, if $f$ is concave in a neighborhood of a point $t$, then the following proposition gives us a sufficient condition under which $f$ is differentiable at this  point. 
\begin{Lem}\label{scf}
Consider a family $(f_\alpha)_{\alpha\in A}$ of concave $C^1(\R)$-functions that are uniformly bounded from above. Define 
\[
f(x)=\sup_{\alpha\in A} f_\alpha(x).
\]
Suppose further that there exist $t\in\R$ and $\eta>0$ such that $f$ is concave on $]t-\eta,t+\eta[$ and   $\alpha^*_t\in A$ such that $f(t)=f_{\alpha^*_t}(t)$. Then, $f$ is differentiable at $t$ with derivative $$f'(t)=f'_{\alpha^*_t}(t).$$ If we suppose moreover that $\alpha^*_t$ is uniquely determined, then $f'$ is continuous at $t$.  
\end{Lem}
 \begin{proof}
By translating the function $f$ if necessary, we can suppose without loss of generality that $t=0$. Because $f$ is 
concave in a neighborhood of $t=0$, we only have to prove that $f'_+(0)\geq f'_-(0)$. To this end, let  $\e>0$ and $\alpha^*_0\in A$ be such that $f(0)=f_{\alpha^*_0}(0)$. Because $f_{\alpha^*_0}$ is concave and differentiable at $0$, for every $\e>0$ there exists $\delta>0$ such that for all $0<h\leq \delta,$ we have 
\begin{equation*}
\frac{f_{\alpha^*_0}(h)-f_{\alpha^*_0}(0)}{h}\geq \frac{f_{\alpha^*_0}(-h)-f_{\alpha^*_0}(0)}{-h}-\e.
\end{equation*}
Thus we get
\begin{align*}
\frac{f(h)-f(0)}{h}
&\geq \frac{f_{\alpha^*_0}(-h)-f_{\alpha^*_0}(0)}{-h}-\e\geq \frac{f(-h)-f(0)}{-h} -\e,
\end{align*}
by the definition of $f$.
Sending $h$ to zero we infer
$
f'_+(0)\geq f'_{\alpha^*_0}(0)\geq f'_-(0)-\e
$
 for every $\e>0$, and hence $f$ is differentiable.
 
 Assume now that $\alpha^*_t$ is uniquely determined, and suppose to the contrary that $f'$ is not continuous at $t$. Since $f$ is concave on $]t-\eta,t+\eta[$ and hence $f'$ is  nonincreasing on $]t-\eta,t+\eta[$, the left- and right-hand limits at $t$ exist, and we infer
 \[
 f'(t^-)=f'_{\alpha_{t^-}^*}(t^-)> f'(t^+)=f'_{\alpha_{t^+}^*}(t^+),
 \]
 where $\alpha^*_{t^-},\alpha^*_{t^+}\in A$. Using the continuity of $f'_{\alpha^*_{t^-}}$  at $t$, 
 we must have, on the one hand,  $\alpha^*_{t^-}\neq\alpha^*_{t^+}$. However, we must equally have, on the other hand,  
 \[
 f(t)=f_{\alpha_t^*}(t)=f(t^+)=f_{\alpha^*_{t^+}}(t^+)=f_{\alpha^*_{t^-}}(t^-),
 \]
 as a direct consequence of the definition of $\alpha^*_t$ and the continuity of $f$.
Therefore, the uniqueness of $\alpha^*_t $ implies $\alpha^*_t=\alpha^*_{t^-}=\alpha^*_{t^+}$, which is clearly a contradiction.
 \end{proof}

We can now state and show  the main result of this subsection.
\begin{Theo}\label{v_r}
The value function is continuously partially differentiable in $R$, and we have the formula 
\[
V_r(T,X,R)=\E\big[u'\big(\cR_T^{\xi^*}\big)\big],
\]
 where $\xi^*$ is the optimal strategy associated to  $V(T,X,R)$.
\end{Theo}
\begin{proof}
The proof is a direct consequence of Lemma \ref{scf}, when applied to the family of concave functions $(R\mapsto \E[u(\cR_T^\xi+R)])_{\xi\in \dot{\X}^1_{2A_2}(T,X_0)}$. Indeed, this is a family of concave $C^1$-functions (due to Proposition \ref{u'}). 
The existence and uniqueness of an optimal strategy (Theorem \ref{eos}) and the concavity of the map $R\mapsto V(T,X,R)$, for fixed $T,X$  (Lemma \ref{cvf}), yield that the remaining conditions of the preceding lemma are satifsfied.
\end{proof}
\begin{Cor}
Suppose that $u'$ is convex and decreasing. Then, the value function is twice differentiable with second partial derivative 
\[
V_{rr}(T,X,R)= \E\big[u''\big(\cR_T^{\xi^*}\big)\big],
\]
 where $\xi^*$ is the optimal strategy associated to  $V(T,X,R)$.
\end{Cor}
\begin{proof}
The proof is similar to the one of Theorem \ref{v_r} and is obtained by applying Lemma \ref{scf} to $u'$ and Proposition \ref{u'}.
\end{proof}
\begin{rem}
We are in the setting of the preceding corollary if, e.g., $u$ is a convex combination of exponential utility functions or, more generally, if 
$(-u)$ is a complete monotone function, i.e., if $\forall n\in\N^*: (-1)^n (-u)^{(n)}\geq0$. According to the Hausdorff-Bernstein-Widder's theorem (cf. \citet{W41} or \citet{D74}, Chapter 21), 
 this is equivalent to the existence of a Borel measure $\mu$ on $[0,\infty[$ such that 
\[
-u(x)= \int_0^\infty e^{-xt}\,d\mu(t).
\]

\xqed{\diamondsuit}
\end{rem}
\subsection{Continuity of the value function}
The proof of the continuity of our value function will be  split  in two propositions. We will  first prove its upper semi-continuity and then its lower semi-continuity. To prove the upper semi-continuity we will use the same techniques as are used to prove the existence of the optimal strategy for the maximization problem \eqref{omp}.
The main idea to prove the lower semi-continuity is to use a convex combination of the optimal strategy for \eqref{omp}  and the optimal strategy of the corresponding exponential value function at a certain well-chosen point. 
Here, we have to distinguish between two cases; the case where the value function is approximated from above, and the case where the value function is approximated from below in time.
In  the sequel, for $\xi\in\dot{\X}^1(T,X_0)$ we will automatically set $\xi_t=0$ for $t\geq T$.
\begin{Prop}\label{usc}
The value function is upper semi-continuous on $]0,\infty[\times\R^d\times \R$.
\end{Prop}
\begin{proof}
Take $\big(T,X_0, R_0\big)\in\;]0,\infty[\times\R^d\times \R$ and let $\big(T^n,X^n_0,R^n_0\big)_n$ be a sequence that converges to $\big(T,X_0,R_0\big)$. We have to show that
\begin{equation}
\limsup_{n} V(T^n,X^n_0,R^n_0)\leq V(T,X_0,R_0).\label{lii}
\end{equation} 

Since $\big(T^n,X^n_0,R^n_0\big)_n$ and $V_i(T^n, X^n_0,R^n_0)$ are bounded, it follows that $\limsup_{n} V(T^n,X^n_0,R^n_0)<\infty$, in conjunction with  \eqref{vfs}. Taking a subsequence if necessary, we can suppose that $(V(T^n,X^n_0,R^n_0))$ converges to $\limsup_{n} V(T^n,X^n_0,R^n_0)$. Let $\xi^{n}$ be the optimal strategy associated to $V(T^n,X^n_0,R^n_0)$, which exists for every $n\in\N,$ due to Theorem \ref{eos}. In the sequel we prove, as in Lemma \ref{osp}, that the sequence $\xi^n$ lies in a weakly sequentially compact set. Note that this proposition can be proved without using Assumption \ref{ar}.
\\
\underline{First step:}  We  set $\widetilde{T}:=\sup_n T^n$. We  will show  that, for every $n\in\N$, we have $\xi^n\in\overline{\cK}_{m}$, provided  that $m$ is large enough, where 
\[
\overline{\cK}_{m}= \Big\{\xi\in \overline{\C}\big(\dot{\X}^1(T^n,X^n_0)\big)_n\big|\; \E\bigg[\int^{\widetilde{T}}_0f(-\xi_t)\,dt\bigg]\leq m\Big\},
\]
and where  $\overline{\C}(\dot{\X}^1(T^n,X^n_0))_n$ denotes the closed convex hull of the sequence of sets $(\dot{\X}^1(T^n,X^n_0))_n$. To this end, we use Remark \ref{km}, noting that  we can choose $\xi^n\in \overline{K}_{m_n},$ where
$m_n$ has to be chosen such that  $$ m_n\geq \frac43\Big(\frac{-V_2(\widetilde{T},X^n_0,R^n_0)}{A_1}+R^n_0 +N\Big),$$ and $N$ depends only on $f,b$ and $\widetilde{T}$. 
Take now  $m\in\R$ such that $m\geq \sup_n m_n.$
Note that such $m$ exists, because $(X_0^n, R^n_0)$ is bounded and $V_2$ is continuous. Then it follows  that $$\E\bigg[\int^{\widetilde{T}}_0f(-\xi^n_t)\,dt\bigg]\leq m \quad\text{ for all } n\in\N.$$ Taking now the convex hull of the sequence of sets $(\dot{\X}^1(T^n,X^n_0))_n$, we  conclude that $\xi^n\in\overline{\cK}_{m}\; \text{ for all } n\in\N$ .\\
\underline{Second step:}  We will prove  that $\overline{\cK}_{m}$ is weakly sequentially compact. To this end, we will  first prove that it is a closed convex set in $\J$. \\The set $\overline{\cK}_{m}$ is convex,  because 
the map $
\xi\longmapsto \E\int^{\widetilde{T}}_0f(-\xi_t)\,dt
$ is convex  (due to the convexity of $f$) and defined on the convex set $\overline{\C}\big(\dot{\X}^1(T^n,X^n_0)\big)_n$. 
We will show that it is closed with respect to the $\J$-norm. Denote by $\overline{\C} (X^n_0)_n$ the closed convex hull of the sequence $(X^n_0)_n$, which is  bounded in $\R^d$. 
We show  that for $\xi\in\overline{\cK}_{m}$ there exists $\widetilde{X}$ in $\overline{\C} (X^n_0)_n$  such that $\xi\in\dot{\X}^1(\widetilde{T},\widetilde{X})$. To this end, we write $\xi$ as a convex combination of $\xi^{n_i}\in\dot{\X}^1(T^{n_i},X^{n_i}_0)$,
\[
\xi=\lambda_1 \xi^{n_1}+\dots+\lambda_s \xi^{n_s},
\]
 where $\sum_{i=1}^{s}\lambda_i=1,\,\lambda_i\geq0$. By expressing then the constraint on $\xi^{n_i}$, we get
\[
\lambda_i \int_0^{T^i}\xi^{n_i}_t\,dt=\lambda_i X_0^{n_i},
\]
which implies 
\[
\int_0^{\widetilde{T}}\xi_t\,dt= \sum_{i=1}^{s} \lambda_i \int_0^{T^i}\xi^{n_i}_t\,dt=\sum_{i=1}^{s}\lambda_i X^{n_i}_0=\widetilde{X}.
\]
Take now a sequence $(\widetilde{\xi}^q)_q$ of $\overline{\cK}_{m}$ that converges in the $\J$-norm to a liquidation strategy $\widetilde{\xi}$. We prove that $\widetilde{\xi}\in\dot{\X}^1(\widetilde{T},\widetilde{X})$ for $\widetilde{X}\in\overline{\C}  (X^n_0)_n$. As previously remarked, there exists a sequence $(\widetilde{X}^q)_q\subset\overline{\C} (X^n_0)_n$ such that  $\widetilde{\xi}^q\in\dot{\X}^1(\widetilde{T},\widetilde{X}^q)$. Hence, we have
\[
\int_0^{\widetilde{T}} \widetilde{\xi}^q\,dt=\widetilde{X}^q,\q
\]
Replacing $(\widetilde{X}^q)_q$ by a subsequence if necessary, we can suppose that it converges to some $\widetilde{X}$, because this sequence is bounded. Moreover,  $\widetilde{X}$ lies in $\overline{\C} (X^n_0)_n$. Since $(\widetilde{\xi}_q)_q$ converges  weakly to $\widetilde{\xi}$, we are now in the setting of Lemma \ref{wcx}, which ensures that $\widetilde{\xi}\in\dot{\X}^1(\widetilde{T},\widetilde{X})$, as well as
$
\E[\int^{\widetilde{T}}_0f(-\widetilde{\xi}_t)]\,\leq m.
$
Hence, this proves that $\overline{\cK}_{m}$ is a closed subset of $\J$. 
\\
Since $\overline{\cK}_{m}$ is convex, it is also closed with respect to the weak topology of $\J$. Thus, it is  sufficient to prove that  $\overline{\cK}_m$ is uniformly integrable. 
To this end, take $\e>0$ and $\xi\in\overline{\cK}_{m}$. There exists  $\alpha>0$ such that $\frac{|\xi_t|}{f(-\xi_t)}\leq \frac \e{m}$, for $\big|\xi_t\big|>\alpha$, due to the superlinear growth property of $f$. Because $f(x)=0$ if and only if $x=0$, the term $1/f(-\xi_t)$ is well-defined on $\{|\xi_t|>\alpha\}$,  hence
\begin{align*}
\E\bigg[\int^T_0\b1_{\{|\xi_t|>\alpha\}}\big|\xi_t\big|\,dt\bigg]
&\leq \E\bigg[\int^T_0\b1_{\{|\xi_t|>\alpha\}}f(-\xi_t)\,dt\bigg] \frac \e{c}
 \leq\e,
\end{align*}
which proves the uniform integrability of $\overline{\cK}_{m}$. 
\\
\underline{Last step:}
We have proved that $(\xi^n)_n$ is a sequence in the weakly sequentially compact set $\overline{\cK}_m$. Thus, there exist a subsequence $\xi^{n_k}$ of $\xi^n$ and some $\widetilde{\xi}\in \overline{\cK}_m$ such that $\xi^{n_k}$  converges to  $\widetilde{\xi}$, weakly in $\J$. We are here again in the settings of  Lemma \ref{wcx}, which allows us us to deduce   that $\widetilde{\xi}\in\dot{\X}^1(T,X_0)$. Finally, because $\xi\longmapsto \E[u(\cR^\xi_T)]$ is upper semi-continuous with respect to the  weak topology of $\J$, due to Proposition \ref{wtc}, we get
\begin{align*}
\limsup_{n} V(T^{n},X^{n}_0,R^{n}_0)
&=\limsup_{k} \E\Big[u\big(\cR^{\xi^{n_k}}_{T}\big)\Big]\leq \E\Big[u\big(\cR^{\widetilde{\xi}}_T\big)\Big]
\leq V(T,X_0,R_0),
\end{align*}
where the last inequality is due to  the definition of $V$ at $(T,X_0,R_0)$ and the fact that $\widetilde{\xi}\in\dot{\X}^1(T,X_0)$. This concludes the proof of the upper semi-continuity of $V$.
\end{proof}
In the following, we will prove the lower semi-continuity of the value function $V$. Contrarily to the proof of the upper semi-continuity of $V$, we will have to consider two cases; when the sequence of time converges from above and from bellow to a fixed time $T$. For the latter case, we will first need to derive  a certain lower semi-continuity property of the value function within time, for fixed $X_0,R_0.$ 
The difficult part of the proof of the lower semi-continuity  is due to the fact that accelerating the strategy when we approximate the time from below  cannot be useful to prove the result, since we are then facing measurability  issues. Therefore we will  have to use other techniques. 
\\
We first need to prove the following lemma, which gives a sufficient condition to ensure that the expected utilities  $\E[u(\cR^{\eta^n}_T)]$ converge to $\E[u(\cR^{\eta}_T)]$, when $\cR^{\eta^n}_T$ converges to $\cR^\eta_T$, in probability.
\begin{Lem}\label{crp}
Let $\eta^n\in\dot{\X}^1(T,X_0)$ be a sequence of strategies such that $\cR^{\eta^n}_T$ converges to $\cR^\eta_T$, in probability, where $\eta\in\dot{\X}^1(T,X_0)$. \\ Suppose moreover that $(\exp(-2A_2 \cR^{\eta^n}_{T^n}))_n$ is uniformly bounded in $\L^2$. Then we have
\begin{equation}
\E\Big[u\Big(\cR_{T}^{\eta^n}\Big)\Big]\underset{n\longrightarrow \infty}{\longrightarrow}  \E\Big[u\Big(\cR_{T}^{\eta}\Big)\Big].
\end{equation}
\end{Lem}
\begin{proof}
We need to prove that $(u(\cR^{\eta^n}_T)_n)$ is uniformly bounded in $\L^2$. But this is a direct consequence of the fact that $(\E[u^+(\cR^{\eta^n}_T)])_n$ is bounded and that, for all $n\in\N$, $\E[(u^-(\cR^{\eta^n}_T))^2]\leq \E[\exp(-2A_2 \cR^{\eta^n}_{T^n})]$, due to inequality \eqref{ubd1}. 
Since $\E[\exp(-2A_2 \cR^{\eta^n}_{T^n})]<\infty$, applying  Vitali's convergence theorem we conclude  that
\begin{equation*}
\E\Big[u\Big(\cR_{T}^{\eta^n}\Big)\Big]\underset{n\longrightarrow \infty}{\longrightarrow}  \E\Big[u\Big(\cR_{T}^{\eta}\Big)\Big].
\end{equation*}
\end{proof}
The next lemma is a direct consequence of the integration by parts formula for the stochastic integral.
\begin{Lem}\label{xsc}
Let $\xi^n\in \dot{\X}^1(T,X_0)$  converge to some $\xi\in\dot{\X}^1(T,X_0)$ in the $\J[0,T]$-weak convergence sense, $\pas$ Then 
\[
\int_0^T (X^{\xi^n}_t)^\top\sigma\,dB_t\underset{n\rightarrow \infty}{\longrightarrow }\int_0^T (X^{\xi}_t)^\top\sigma\,dB_t \quad \pas
\]
 \end{Lem}
Now we are ready to state and prove the following proposition. 
\begin{Prop}\label{cfbt}
Let $(T,X_0, R_0)\in\;]0,\infty[\times\R^d\times \R$ and $T^n$ be a sequence of positive real numbers that converges from below to $T$, i.e., $T^n\uparrow T$. Then we have
\begin{equation}
\liminf_{n} V(T^n, X_0,R_0)\geq V(T,X_0,R_0).\label{iva}
\end{equation}
\end{Prop}
\begin{proof} In the following, we will need Assumption \ref{ar}. Let $(T,X_0,R_0)\in\,]0,\infty[\times\R^d\times\R$ and $\xi\in\dot{\X}^1_{2A_2}(T,X_0)$.
Define 
\begin{align*}
\varphi^\xi:&\,]0,\infty[\longrightarrow \R\\
&\overline{T}\longmapsto \E\big[u\big(\cR^\xi_{\overline{T}}\big)\big].
\end{align*}
Note that the map $\varphi^\xi$ is constant on $[T,\infty[$. We show  that $\varphi^\xi$ is continuous at $T$. To this end, it is sufficient to take a sequence
$(T^n)$ such that $T^n\uparrow T$ and to prove that 
\begin{equation}
\varphi^\xi(T^n)\longrightarrow \varphi^\xi(T)\label{vcv}
\end{equation}
or, equivalently,  
\begin{equation*}
\E\big[u\big(\cR^\xi_{T^n}\big)\big]\longrightarrow \E\big[u\big(\cR^{\xi}_T\big)\big].
\end{equation*}
We easily have the   convergence 
\begin{equation}
\cR_{T^n}^{\xi}=\int_0^{T^n} (X^{\xi}_{ t})^\top\sigma\, dB_t+\int_0^{T^n} b\cdot X^{\xi}_t\,dt-\int_0^{T^n} f(-\xi_t)\,dt\underset{n\rightarrow \infty}{\longrightarrow} \cR^{\xi}_T\q 
\end{equation}
Because $u$ is continuous, we then obtain
\begin{equation}
\lim_nu\big(\cR_{T^n}^{\xi}\big)= u\big(\cR_{T}^{\xi}\big)\quad \pas
\end{equation} 
Now, we have to  prove the boundedness of the sequence $(\E[\exp(-2A \cR^{\xi}_{T^n})])_n$.
For this matter, we write
\begin{IEEEeqnarray*}{rCl} 
\IEEEeqnarraymulticol{3}{l}{\E\big[\exp\big(-2A \cR^{\xi}_{T^n}\big)\big]}\\
&\leq& K\E\Big[\exp\Big(-2A \Big(\E\bigg[\int_0^{T} (X^{\xi}_{ t})^\top\sigma\, dB_t+\int_0^{T} b\cdot X^{\xi}_t\,dt-\int_0^{T} f(-\xi_t)\,dt\Big|\F_{T^n}\Big]\Big)\Big)\Big]\\
&\leq& K\E\Big[\E\Big[\exp\Big(-2A \Big(\int_0^{T} (X^{\xi}_{ t})^\top\sigma\, dB_t+\int_0^{T} b\cdot X^{\xi}_t\,dt-\int_0^{T} f(-\xi_t)\,dt\Big)\Big)\Big|\F_{T^n}\Big]\Big]\\
 &=&K\E\Big[\exp\Big(-2A \Big(\int_0^{T} (X^{\xi}_{ t})^\top\sigma\, dB_t+\int_0^{T} b\cdot X^{\xi}_t\,dt-\int_0^{T} f(-\xi_t)\,dt\Big)\Big)\Big]
<\infty,
 \end{IEEEeqnarray*}
 where $K=\exp (T|b|\|X^{\xi}\|_{\L^2})$ is obtained  using H\"older's inequality, and where the finiteness of the last term follows with $\xi\in\dot{\X}^1_{2A_2}(T,X_0)$.
 Thus, the sequence   $(u(\cR^{\xi}_{T^n})$ is uniformly bounded in $\L^2$, whence  using Vitali's convergence theorem we infer
  \[
\E\big[u\big(\cR_{T^n}^{\xi}\big)\big]\underset{n\rightarrow \infty}{\longrightarrow}  \E\big[u\big(\cR_{T}^{\xi}\big)\big],
\]
which proves \eqref{vcv}. Hence, $\varphi^\xi$ is continuous at $T$, and $\sup_{\xi\in\dot{\X}^1_{2A_2}(T,X_0)}\varphi^\xi$ is lower semi-continuous at $T$, because it is the supremum of a family of (lower semi-) continuous functions. Since $$\sup_{\xi\in\dot{\X}^1_{2A_2}(T,X_0)}\varphi^\xi (T)=V(T,X_0,R_0),$$ this proves in particular that for every 
 sequence of time $T^n$ that converges from below to $T$,  we have
\begin{equation}
\liminf_{n} \sup_{\xi\in\dot{\X}^1_{2A_2}(T,X_0)}\varphi^\xi(T^n)\geq \sup_{\xi\in\dot{\X}^1_{2A_2}(T,X_0)}\varphi^\xi(T)=V(T,X_0,R_0),
\end{equation}
which proves \eqref{iva}.
\end{proof}
We can now derive the lower semi-continuity of the value function $V$.
\begin{Prop}\label{lsc}
The value function is lower semi-continuous on $]0,\infty[\times\R^d\times \R$.
\end{Prop}
\begin{proof}
Let $(T,X_0, R_0)\in\;]0,\infty[\times\R^d\times \R$ and $(T^n,X^n_0,R^n_0)_n$ be a sequence that converges to $(T,X_0,R_0)$. We have to show that
\begin{equation}
\liminf_{n} V(T^n, X^n_0,R^n_0)\geq V(T,X_0,R_0).\label{lsi}
\end{equation}
We split the proof of \eqref{lsi} in two parts; first we will assume that $T^n\downarrow T$, second we will assume that $T^n\uparrow T$ (for this latter case, we will use Proposition \ref{cfbt}).\\
\underline{First case:} Suppose that $T^n\downarrow T$. 
We  set  
\begin{equation}
\lambda_n:=\begin{cases}|X^n_0-X_0\big|,& \text{if } |X^n_0-X_0|\neq0,\\
                                               \frac1{n},&\text{otherwise},\\
                          \end{cases}\label{lan}
\end{equation}
                           which belongs to $]0,1[$, for $n$ large enough. Let  now $\widehat{X}^n_0\in\R^d$ be  such that
$
X^n_0=(1-\lambda_n)X_0+\lambda_n \widehat{X}^n_0
$
 and consider  the  sequence of strategies
\begin{equation*}
\xi^n_t:=(1-\lambda_n) \xi^*_{t}+\lambda_n\widehat{\xi}^n_t,
\end{equation*}
 where  $\xi^*$ is the optimal strategy associated to  $V(T,X_0,R_0)$, and $\widehat{\xi}^n$ is the optimal strategy associated to  $V_2(T^n, \widehat{X}^n_0, R^n_0)$. \\ 
 Note that, due to the choice of $\lambda_n$, the vector $\widehat{X}^n_0$ is bounded: indeed, we have
 \[
\widehat{X}^n_0=\frac{X^n_0-X_0}{\lambda_n}+\lambda_n +X_0,
\]
which is bounded, due to the boundedness of $X^n_0$ and the definition of $\lambda_n$.
Hence, $V_2 (T^n,\widehat{X}^n_0,R^n_0)$ is bounded in $n$, which implies that $\int^{T^n}_0 f(-\widehat{\xi}^n_t)\,dt$ 
is again bounded in $n$. 
Since $f$ has superlinear growth and is positive, the integral $\int^{T^n}_0|-\widehat{\xi}^n_t|\,dt$ is also bounded in $n$.

Observe that
\begin{align*}
\int^{T^n}_0 \xi^n_t\;dt= (1-\lambda_n)\int^{T^n}_0 \xi^*_{t}\;dt +\lambda_n \int^{T^n}_0\widehat{\xi}^n_t\;dt=(1-\lambda_n)X_0 + \lambda_n \widehat{X}^n_0=X^n_0,
\end{align*}
where the last equality follows  with $T^n\geq T$ and the fact that $\xi^*_t=0$ for $t\geq T$. Moreover, $\xi^n$ verifies \eqref{ias}, due to the convexity of $f$ and the boundedness of $\widehat{\xi}^n$,
whence  $\xi^n\in \dot{\X}^1_{2A_2}(T^n,X^n_0)$.

 We now show  that
\begin{equation}
\cR_{T^n}^{\xi^n}=\int_0^{T^n} (X^{\xi^n}_{ t})^\top\sigma\, dB_t+\int_0^{T^n} b\cdot X^{\xi^n}_t\,dt-\int_0^{T^n} f(-\xi^n_t)\,dt\underset{n\rightarrow \infty}{\longrightarrow} \cR^{\xi^*}_T,\q,\label{rcn}
\end{equation}
by individually consedering each term, starting from the left.\\
 Because $\int^{T^n}_0|\widehat{\xi}_t^n|\,dt$ is uniformly bounded, $\xi^n$ converges to $\xi^*$ in $\J[0,T],\;\pas$ Indeed, we write
\begin{align*}
\E\bigg[\int_0^{T^n}\big|\xi^n_t-\xi^*_t\big|\, dt\bigg]
&=\lambda_n \Big(\E\bigg[\int_0^{T^n}\big|\widehat{\xi}^n_t\big|\, dt\bigg]+\E\bigg[\int_{T^n}^T\big|\xi^*_t\big|\, dt\bigg]\Big)
\underset{n\rightarrow \infty}{\longrightarrow} 0.
\end{align*}
Therefore, Lemma \textcolor{red}{\ref{xsc}} yields
\[
\int_0^{T^n} (X^{\xi^n}_{ t})^\top\sigma\, dB_t\underset{n\rightarrow \infty}{\longrightarrow} \int_0^T(X^{\xi^*}_{ t})^\top\sigma\, dB_t.
\] 
Due to $X^{\xi^n}_t=(1-\lambda_n)X^{\xi^*}_{ t}+ \lambda_n X^{\widehat{\xi}^n}_t \; \pas \;\text{ for all } t\in[0,T^n]$, we can express the second integral in \eqref{rcn} as follows:
\begin{align*}
\int_0^{T^n} b\cdot X^{\xi^n}_t\,dt
&=(1-\lambda_n)\int_0^{T} b\cdot X^{\xi^*}_{ t} \, dt + \lambda_n\int_0^{T^n} b\cdot X^{\widehat{\xi}^n}_{ t} \, dt,
\end{align*}
which converges $\pas$ to $\int_0^{T} b\cdot X^{\xi^*}_{ t} \, dt$, because $\int_0^{T^n} b\cdot X^{\widehat{\xi}^n}_{ t} \, dt$ is uniformly bounded and $\lambda_n$ is a null sequence.

  We now prove  that
\begin{equation}
\int_0^{T}  f\big(-(1-\lambda_n) \xi^*_t-\lambda_n\widehat{\xi}^n_{t}\big)\,dt\underset{n\rightarrow \infty}{\longrightarrow}\int_0^{T} f(-\xi^*_{t})\,dt,\quad \pas\label{fc}
\end{equation}
Due to the continuity of $f$, we have
\begin{equation*}
 f\big(-(1-\lambda_n) \xi^*_t-\lambda_n\widehat{\xi}^n_{t}\big)\longrightarrow f\big(-\xi^*_t\big), \quad\pas
\end{equation*}
Because $f$ is convex, we further get
\[
0\leq f\big(-(1-\lambda_n) \xi^*_t-\lambda_n\widehat{\xi}^n_{t}\big)\leq (1-\lambda_n) f\big(-\xi^*_t\big)\,dt+ \lambda_n  f\big(-\widehat{\xi}^n_t\big).
\]
Since $\int_0^Tf(-\widehat{\xi}^n_t)\,dt$ is uniformly bounded in $n$,  the dominated convergence theorem of Lebesgue implies  \eqref{fc}. Therefore, \eqref{rcn} is established,
whence again
\begin{equation}
\lim_n u\big(\cR_{T^n}^{\xi^{n}}\big)= u\big(\cR_{T}^{\xi^{*}}\big)\quad \pas,\label{ule}
\end{equation} 
using the continuity of $u$.

Further, with $L:=\sup_n V_2(T^n, \widehat{X}^n_0, R^n_0)$, we obtain
\begin{align*}
\exp(-2A_2\cR_{T^n}^{\xi^n}) &\leq  \big((1-\lambda_n) \exp(-2A_2\cR_{T^n}^{\xi^*})+\lambda_n \exp(-2A_2\cR_{T^n}^{\widehat{\xi}^n})\big)\\ 
&\leq  \big((1-\lambda_n)M_{\cR_{T}^{\xi^*}}(2A2) +\lambda_n L\big)<\infty,
\end{align*}
because $\xi\mapsto\exp(-2A \cR_{T^n}^\xi)$ is convex and $T^n\geq T$, in conjunction with  Assumption \ref{ar}.    Therefore,  applying Lemma \ref{crp} gives
\[
\E\big[u\big(\cR_{T^n}^{\xi^n}\big)\big]\underset{n\longrightarrow \infty}{\longrightarrow}  \E\big[u\big(\cR_{T}^{\xi^*}\big)\big].
\]
Finally, we can write
\begin{align*}
\liminf_n V(T^n, X^n_0,R^n_0)\geq \liminf_{n} \E\left[u\left(\cR_{T^n}^{\xi^n}\right)\right]= \E\left[u\left(\cR_{T}^{\xi^*}\right)\right]
=V(T,X_0,R_0),
\end{align*}
which proves \eqref{lsi} when  $T^n\downarrow T$. 
\\
\underline{Second case:} Suppose now that $T^n\uparrow T$. We  let  $\lambda_n$ and  $\widehat{X}^n_0\in\R^d$ as in \eqref{lan} and consider the following sequence of strategies
\begin{equation*}
\xi^n_t:=(1-\lambda_n) \xi^{*,n}_{t}+\lambda_n\widehat{\xi}^n_t,
\end{equation*}
 where  $\xi^{*,n}$ is the optimal strategy associated to  $V(T^n,X_0,R_0)$ and  $\widehat{\xi}^n$  is the optimal strategy associated to $V_2(T^n, \widehat{X}^n_0, R^n_0)$.
 
As above, we can  show that $\xi^n\in \dot{\X}^1_{2A_2}(T^n,X^n_0)$, wherefore
\begin{align*}
\liminf_n V(T^n, X^n_0,R^n_0)&\geq 
 \liminf_{n} \E\big[u\big(\cR_{T^n}^{(1-\lambda_n) \xi^{*,n}+\lambda_n\widehat{\xi}^n}\big)\big]\\
&\geq\liminf_{n} \big( (1-\lambda_n)\E\big[u\big(\cR_{T^n}^{\xi^{*,n}}\big)\big]+\lambda_n \E\big[u\big(\cR_{T^n}^{\widehat{\xi}^n}\big)\big]\big)\\
&\geq  \liminf_{n} (1-\lambda_n)V(T^n,X_0,R_0)+\liminf_{n}\lambda_n V_2(T^n,X^n_0,R^n_0)\\
&\geq V(T,X_0,R_0).
\end{align*}
Here, we have used the concavity of $\xi\mapsto \E[u(\cR_{T}^{\xi})]$ for the second inequality, inequality \eqref{vfs} for the third one, and Proposition \ref{cfbt}, in conjunction with the fact that $V_2(T^n,X^n_0,R^n_0)$ is bounded and $\lambda_n$ is a null sequence, for the last one.
This proves \eqref{lsi} when  $T^n\uparrow T$. 
\end{proof}
As a consequence of Proposition \ref{usc} and Proposition \ref{lsc}, we obtain the following fundamental result.
\begin{Theo}\label{cv}
The value function $V$ is continuous on $]0,\infty[\times\R^d\times \R$. 
\end{Theo}
\subsection{The Bellman principle and the construction of $\e$-maximizers.}
In this section we prove the Bellman principle of optimality underlying our maximization problem \eqref{omp}. To this end, we use  
$\e$-maximizers  constructed on a bounded region. Their existence is proved by using an approximating sequence of strategies. Thus, we avoid here the use of a measurable selection theorem, which appears typically in optimal control theory. The dynamic programming principle is a key result to prove both a verification theorem and a theorem stating that the value function is a solution, in the viscosity sense, 
of a Hamilton-Jacobi-Bellman equation. From now on, for a fixed time $T\in\;]0,\infty[$, we will consider the time-reversed value function: $t\mapsto V(T-t,X_0,R_0)$,  
and we will assume that $(\Omega, \F,\P)$ is the \emph{canonical Wiener Space}. 
\begin{Theo}
\emph{(Bellman Principle)}
\label{bp}
Let $(T,X_0,R_0)\in\;]0,\infty[\times\R^d\times\R$. Then we have
 \begin{equation}
 V(T, X_{0}, R_{0})=\sup_{\xi\in\dot{\X}^1(T, X_0)}\E\big[V\big(T-\tau, X^\xi_{\tau}, \cR^{\xi}_{\tau}\big)\big]
\label{ebp}
\end{equation}
for every stopping time $\tau$ taking values in $[0,T[$.
\end{Theo} 
\begin{rem}
Note that \citet{BT11} developed  a weak formulation of the dynamic principle, which can be used to derive the viscosity property of the corresponding value function, in some optimal control problems. However, this requires the following concatenation property (\emph{Assumption A}) of the strategies: for $\xi,\eta\in\dot{\X}^1(T,X_0)$ and a stopping time $\tau\in[0,T[$, we must have that $\xi\b1_{[0,\tau]}+\eta\b1_{]\tau,T]}\in\dot{\X}^1(T,X_0)$, which is however not the case in general, and therefore is not usable in our work. In \citet{BN12}, another weak formulation of the dynamic principle with generalized state constraints is formulated. Here again, a concatenation property (\emph{Assumption B}) in the following form is required: for $\xi,\eta\in\dot{\X}^1(T,X_0)$ and a time $s\in[0,T]$, it must hold that $X^\xi_t=X^\xi_s-\int_s^t\eta_u\, du$, for $t\leq s$, which is again not the case in general, and thus cannot be directly applied here.\xqed{\diamondsuit}
\end{rem}
 The proof of Theorem \ref{bp} is split in two parts. For ease of reference, let us first make the following assumption  on $f$.
 \begin{Ass}\label{afp}
From now on, we suppose that $f$ has at most a polynomial growth of degree $p$, i.e.,  there exists $C>0$ such that $$ f(x)\leq C(1+|x|^p)\quad \text{ for all }x\in\R^d.$$
\end{Ass}
Further, in order to avoid measurability issues, we need to suppose that for $T\in\;]0,\infty[,\;(\Om, \F, (\F_t)_{t\in[0,T]}, P)$ is the canonical Wiener space. Taking this perspective, let us start with proving some measurability results. Here also, we will restrict our attention to  strategies that lie in $\dot{\X}^1_{2A_2}(T, X_0,R_0)$, as mentioned in Assumption \ref{ar}. 
\begin{Lem}\label{ce} 
 For $\om\in\Om,$ define the map $\phi_\om:\Om \rightarrow \Om$ by
\begin{equation*}
\phi_\om(\widetilde{\om})=\begin{cases}
\om(s),& \text{ for } s\in [0,\tau(\om)],\\
\om(\tau(\om))+\widetilde{\om}(s)-\widetilde{\om}(\tau(\om)),&\text{ for } s\in\;]\tau(\om), T],
\end{cases} 
\end{equation*}
where $\tau$ is as in \eqref{ebp}. Moreover, for $\xi\in\dot{\X}^1(T,X_0)$ we define 
\[
\xi^\om_t(\widetilde{\om}):=\xi_t\circ\phi_\om(\widetilde{\om}).
\]
Then, for $\P$-a.e. $\om$,
\begin{equation}
\E\left[u\big(\cR_T^\xi\big)\big|\F_\tau\right](\om)=\E\left[u\big(\cR_{\tau}^\xi+\cR_{\tau, T}^{\xi^{\om}}\big)\big|\F_{\tau}\right](\om)=\E\Big[u\big(\cR_{\tau}^{\xi}(\om)+\cR_{\tau(\om),T}^{\xi^{\om}}\big)\Big],\label{ma}
\end{equation} 
 where
$R_{t, T}^{\widetilde{\xi}}$ denotes the revenues generated by the strategy $\xi^{\om}$ during the time period $[t,T]$, i.e:
\begin{equation*}
R_{t, T}^{\widetilde{\xi}}=  \int_{t}^{T}(X^{\widetilde{\xi}}_{s})^\top\sigma\,dB_{s}+\int_{t}^{T} b\cdot X^{\widetilde{\xi}}_{s}\;ds-\int_{t}^{T}f(-\widetilde{\xi}_{s})\;ds.
\end{equation*}
\end{Lem}
To prove the preceding Lemma, we have to use the three following lemmas. The proof of the first one can be found in, e.g., \citet{YR91} (as a consequence of Levy's characterization of Brownian motion) or \citet{HK04}.
\begin{Lem}
Let $\tau$ be a bounded stopping time and $(B_t)_{t\in[0,\infty[}$
a Brownian motion. Then $\widetilde{B}_t:=B_{t+\tau}-B_{\tau}$ is a Brownian motion independent of $\F_{\tau}$.
\end{Lem}

The next lemma uses the Dynkin's $\pi$-$\lambda$ theorem. See, e.g.,  \citet{W91} for more details.
\begin{Lem}
Let $F:\R^2\longrightarrow [0,\infty[$ be a measurable function, $X$ independent of a sigma-algebra $\A$ and $Y\;\A$-measurable. Then,
\begin{equation}
\E[F(X,Y)\big|\A](\om)=\E[F(X,Y(\om))]\q \label{cee}
\end{equation}
\end{Lem}
\begin{proof}
Let us first consider $A=(A_1\times A_2), A_i\in\B(\R), i=1,2,$ and set 
\[
F(x,y):=\b1_{A_1\times A_2}(x,y)=\b1_{A_1}(x)\b1_{A_2}(y).
\]
Using the fact that $Y$ is $\A$-measurable 
as well as the independence of $X$ 
we write 
\begin{align*}
\E[F(X,Y)](\om)&=\E[\b1_{A_1}(X)\b1_{A_2}(Y)\big|\A](\om)\\
&=\b1_{A_2}(Y(\om))\E[\b1_{A_1}(X)\big|\A](\om)\\
&=\b1_{A_2}(Y(\om))\E[\b1_{A_1}(X)]\\
&=\E[\b1_{A_1}(X)\b1_{A_2}Y(\om)].
\end{align*}
 Consider now 
\[
\D:=\{ A\in\B(\R^2)\,\big|\eqref{cee}\text{ holdsÊfor } F=\b1_A\}.
\] 
Then $\D$ is a Dynkin system  containing $\C:=\{A_1\times A_2\big|A_i\in\B(\R)\}$. Due to the stability of the set $\C$ under intersection, it follows that $\D\supset \sigma(\C)=\B(\R^2)$. Using the monotone convergence theorem,
\eqref{cee} follows for an arbitrary $F$.
\end{proof}
The next lemma is a consequence of both preceding results.
\begin{Lem}
Let $H:\Om\longrightarrow [0,\infty[$ be a measurable function, $\tau$  a stopping time with values in $[0,T[$, and $\phi_w$  defined as in Lemma \ref{ce} for $\om\in\Om$. Then we have
\begin{equation*}
\E[H\big|\F_\tau](\om)=\E[H\circ\phi_\om]\q 
\end{equation*}
\end{Lem}
We can now prove Lemma \ref{ce}
\begin{proof}[Proof of Lemma \ref{ce}]
First, note that
\begin{align*}
\cR^\xi_T\circ\phi_\om(\widetilde{\om})&=\cR^\xi_\tau\circ\phi_\om(\widetilde{\om})+\cR^\xi_{\tau,T}\circ\phi_\om(\widetilde{\om})\\
&=\cR^\xi_\tau(\om)+\cR^{\xi^\om}_{\tau(\om),T}(\widetilde{\om})
\end{align*}
for $\P$ -a.e. $\widetilde{\om}\in\Om$. Due to the fact that $u$ is bounded from above, we can apply the preceding Lemma to $H:= -u(\cR^\xi_T)$ (by translating $u$ vertically if necessary),  and we finally get (when dropping the minus sign in front of $u$)
\begin{align*}
 \E\big[u\big(\cR_T^\xi\big)\big|\F_\tau\big](\om)&= \E\big[u\big(\cR_T^\xi\circ\phi_\om\big)\big]\\
&= \E\big[u\big(\cR_{\tau}^{\xi}(\om)+\cR_{\tau(\om),T}^{\xi^{\om}}\big)\big],
\end{align*}
which proves the lemma.
\end{proof}
The following lemma yields an upper bound for an exponential value function at some stopping time with values in $[0,T[$. It uses the notations of Lemma \ref{ce}. For $d=1$, an analogous result  can  be found in \citet{SSB08}.
\begin{Lem}\label{lice}
Let $\overline{V}(T,X_0,R_0)=\inf_{\xi\in\dot{\X}_{det}(T,X_0)} \E\big[\exp(-A \cR_T^\xi)\big]$ and $\tau$ be a stopping time with values in $[0,T[$. We then have
\begin{equation}
\overline{V}(T-\tau, X_\tau^\zeta, \cR_\tau^\zeta)\leq \E\big[\exp(-A\cR_T^\zeta)|\F_\tau\big]\label{mp}\q
\end{equation}
for every $\zeta \in\dot{\X}^1(T,X_0).$
\end{Lem}
\begin{proof}
Let $\tau\leq T$ be a stopping time, $\zeta \in\dot{\X}^1(T,X_0)$, and denote by
\begin{equation}
\cR_{s, T}^\zeta=    \int_s^T ( X^\zeta_t)^\top \sigma\;dB_t +\int_s^T b\cdot X^\zeta_t \;dt-\int_s^T f(-\zeta_t)\;dt\label{rsT}
\end{equation}
the revenues generated by $\zeta$ over the time interval $[s,T]$.
 In \cite{SST10}, there is another convenient formulation of $\overline{V}$: for every $\om\in\Om$,
\[
\overline{V}(T-\tau(\om),X_{\tau}^\zeta (\om),\cR_{\tau}^\zeta (\om))=\exp\Big(-A \cR^\zeta_\tau (\om)+ A \inf_{\widetilde{\zeta}\in\dot{\X}_{det}(T-\tau(\om), X_\tau^\zeta(\om))}\int_\tau^T\cL(X_t^{\widetilde{\zeta}},\widetilde{\zeta}_t)\,dt\Big).
\]
Let us next set
\[
Y^\zeta=e^{-A \int_\tau^T  ( X^\zeta_t)^\top \sigma\;dB_t-\frac12\int_\tau^T A^2(X^\zeta_t)^\top \Sigma X_t^\zeta\,dt}.
\]
We then have for every $\zeta\in\dot{\X}^1(T,X_0)$ and almost every $\omega\in\Omega$: 
\begin{align*}
\lefteqn{\E\Big[\exp(-A \cR_{\tau,T}^\zeta)|\F_\tau\Big](\omega)} & \\
&= \E\bigg[Y^\zeta\exp\bigg(A\int_\tau^T \cL(X_t^\zeta,\zeta_t)\,dt\bigg)\bigg|\F_\tau\bigg](\om)\\
&\geq \E\bigg[Y^\zeta\exp\bigg(A\inf_{\widetilde{\zeta}\in\dot{\X}_{det}(T-\tau(\om), X_\tau^\zeta(\om))} \int_\tau^T \cL(X_t^{\widetilde{\zeta}},\widetilde{\zeta}_t)\,dt\bigg)\bigg|\F_\tau\bigg](\om)\\
&= \E\Big[Y^\zeta e^{A \cR^\zeta_\tau (\om)}\overline{V}(T-\tau(\om),X_{\tau}^\zeta (\om),\cR_{\tau}^\zeta (\om) |\F_\tau\Big](\om)\\
&=\exp\big(A \cR^\zeta_\tau (\om)\big)\overline{V}(T-\tau(\om),X_{\tau}^\zeta (\om),\cR_{\tau}^\zeta (\om))\E\big[Y^\zeta  |\F_\tau\big](\om).
\end{align*}
Here, we have used \eqref{rsT} for the first equality and  the monotonicity property of the conditional expectation for the inequality.

It remains to show  that
\begin{equation}
\E\big[Y^\zeta  |\F_\tau\big]=1\label{y|}\q.
\end{equation}
Indeed, this will prove the result, because we also have that
\begin{align*}
\E\big[\exp(-A\cR_T^\zeta)|\F_\tau\big](\omega)&= \E\big[\exp\big(-A\big(\cR_{\tau, T}^{\zeta}+ \cR_\tau^\zeta(\omega)\big)\big)\big|\F_\tau\big](\om)\\
&=   \exp\big(-A \cR_\tau^\zeta(\omega)\big) 
\E\big[\exp\big(-A \cR_{\tau, T}^{\zeta}(\om)\big)\big|\F_\tau\big](\om),
\end{align*}
by using  \eqref{ma}.
To prove \eqref{y|},  let us define the following process 
\[
Z^{\zeta}_t=e^{-A \int_0^t  ( X^\zeta_u)^\top \sigma\;dB_u-\frac12\int_0^t A^2(X^\zeta_u)^\top \Sigma X_u^\zeta\,du},
\]
which is a true martingale, due to Girsanov's theorem ($X^\zeta$ fulfills \eqref{ias}, due to the assumption on $\zeta$). Therefore, we have
\begin{align*}
\E\big[Z_T^\zeta|\F_\tau\big]&=\E\big[Y^\zeta Z_\tau^\zeta|\F_\tau\big]  \\
&=Z_\tau^\zeta \E\big[Y^\zeta|\F_\tau\big]\\
&=Z_\tau^\zeta,
\end{align*}
which proves \eqref{y|} and hence also our lemma.
\end{proof}
We wish now to prove the following fundamental proposition:
\begin{Prop}\label{eiv}
Let $ \xi\in\dot{\X}^1_{2A_2}(T,X_0)$ and $\tau$ be a stopping time with values in $[0,T[$. Then we have
\begin{equation}
V(T,X_0,R_0)\geq \E\big[V\big(T-\tau, X^\xi_\tau, \cR^\xi_\tau\big)\big]. \label{eive}
\end{equation}
\end{Prop}
This proposition will follow from the subsequent lemma and the theorem on the existence of  $\e$-maximizers on a bounded region. The latter one will be proved without the use of a measurable selection argument, by simply using the continuity of the value function and the existence of an optimal strategy for the maximization problem \eqref{omp}. 
The next lemma allows us to restrict our problem to a region where the parameters $T, X_0$ and $R_0$ are bounded. Indeed, outside this region (with the bound of the parameters having to be taken large enough), the following result proves that the right-hand side term of $\eqref{eive}$ can be chosen
smaller than $\e$.
\begin{Lem}\label{eN}
Let $\xi\in\dot{\X}^1_{2A_2}(T,X_0)$. Under the assumptions and notations of Proposition \ref{eiv}, there exists  $N=N_\e\in\N$ such that
\begin{equation}
\E\Big[\big|V(T-\tau,X^\xi_\tau,\cR^\xi_\tau)\big|\b1_{\big\{\left|X^\xi_\tau\right|\vee\left|\cR^\xi_\tau\right|> N\big\}}\Big]\leq\e.\label{en}
\end{equation}
\end{Lem}
\begin{proof}
We first prove that
\begin{equation}
\E\big[|V_2(T-\tau, X^\xi_\tau, \cR^\xi_\tau)|\big]<\infty\label{cb},
\end{equation}
where we  have $|V_2(T,X_0,R_0)|=\inf_{\zeta\in\dot{\X}(T,X_0)} \E\big[\exp(-A_2\cR_T^\zeta)\big]$. This is a direct consequence of Lemma \ref{lice}. Indeed, we can write
\begin{align*}
\E\big[|V_2(T-\tau, X_\tau^\xi, \cR_\tau^\xi)|\big]&\leq \E\big[\E\big[\exp(-A_2\cR_T^\xi)|\F_\tau\big]\big]\\
&=\E\big[\exp(-A_2\cR_T^\xi)\big]\\
&<\infty.
\end{align*}
Here, the first inequality is due to \eqref{mp}, and the  last one follows from the fact that $\xi\in\dot{\X}^1_{2A_2}(T,X_0)$. Thus \eqref{cb} follows, and hence, there exists  $N\in\N$ such that
\begin{equation*}
\E\Big[\big(|V_2(T-\tau,X^\xi_\tau,\cR^\xi_\tau)|+1/A_1\big)\b1_{\big\{\left|X^\xi_\tau\right|\vee\left|\cR^\xi_\tau\right|> N\big\}}\Big]\leq\e.
\end{equation*}
Using 
\[
|V(T,X_0,R_0)|\leq |V_2(T,X_0,R_0)|+1/A_1, \quad(T,X_0,R_0)\in\;]0,\infty[\times\R^d\times\R,
\]
which is due to \eqref{vfs}, we infer \eqref{en}.
\end{proof}
We can now state and prove the following fundamental theorem of this subsection.
\begin{Theo}[Existence of the $\e$-maximizers on a bounded region]\label{em}
With the notations of Proposition~\ref{eiv}, Lemma \ref{ce} and Lemma \ref{eN}, there exists a progressively measurable process $\widetilde{\xi}^.=\widetilde{\xi}^{.,\tau,\e}\in\dot{\X}^1_{2A_2}(T-\tau(.), X_\tau^\xi(.))$ such that for $\mathrm{P}$-a.e.      $\om\in\big\{\big|X^\xi_\tau\big|\wedge\big|\cR^\xi_\tau\big|\leq N\big\}$,
\begin{equation}
V\big(T-\tau(\om),X^\xi_\tau(\om),\cR^\xi_\tau(\om)\big)\leq \E\Big[u\Big(\cR^\xi_\tau(\om) +\cR^{\widetilde{\xi}^{\om,\tau,\e}}_{\tau(\om),T}\Big)\Big]+\e. \label{emi}
\end{equation}
\end{Theo}
\begin{proof}
The proof of this result is split in several steps. Let us first consider a simple process $\xi$  which is allowed to take only countably many values and a discrete stopping time $\tau$. The existence of the $\e$-maximizers is  easier to prove in this case, because we are not facing any measurability problems.

In the second step, we consider an arbitrary process $\xi\in\dot{\X}^1_{2A_2}(T,X_0)$ and a stopping time $\tau$ taking values in $[0,T[$. The process $\xi$ can then be approximated by simple processes as in the first step,
 with respect to the  topology of the $\L^p$-norm, where $p$ has to be chosen such that $f(x)\leq C(1+|x|^p)$ (see Assumption \ref{afp}).
 
 In the third step, we show by compactness arguments that the corresponding sequence of $\e$-maximizers (as obtained in the first step) converges weakly to a process $\xi^{\tau,\e}$.

In the last step, we show that $\xi^{\tau,\e}$ is the $\e$-maximizer we were looking for. \\
As observed in  Remark  \ref{km}, we will use the fact that a process $\xi\in\dot{\X}^1_{2A_2}(T,X_0)$ lies, in particular, in the set $\overline{K}_m (T,X_0)$ for a constant $m>0$, with
\[
\overline{K}_{m}(T,X_0)=\Big\{\xi\in \dot{\X}^1(T,X_0)\big|\; \E\bigg[\int^T_0f(-\xi_t)\,dt\bigg]\leq m\Big\}.
\]
\\
\underline{First step:}  Let $\e>0$. For $L\in\N$ and $i\in\{0,\dots,2^{L}\}$, define
\[
t_i=i \frac{T}{2^{L}},
\]
and $\xi\in\dot{\X}^1_{2A_2}(T,X_0)$ as follows:
\begin{equation}
\xi_t(\om)=\sum_{i=1}^{2^{L}}\xi_i(\om)\b1_{[t_i,t_{i+1}[}(t),
\end{equation}
where $\xi_i$ takes values in the set $\{z_{i,p}\,|\,p\in\N, z_{i,p}\in\R^d\}$.
Moreover, let $\tau$ be a stopping time taking values in the set $\{t_0,t_1,...,t_{2^L}\}$, and set $\Om_{i,p_i}:=\{\xi_i=z_{i,p_i}\}$, $\Gamma_j:=\{\tau=t_j\}$. Note that $\Gamma_j$ and $\Om_{i,p_i}$ can be empty. For every $t\in [0,T]$, we have
\begin{equation}
X^\xi_t=X_0-\sum_{i=1}^{k-1} \xi_i (t_{i+1}-t_i) - \xi_{k}(t-t_{k}),
\end{equation}
where $k$ is such that $t\in [t_k,t_{k+1}[$.  We can therefore write for every $\om\in\bigcap_{i=1}^q\Om_{i,p_i}\cap \Gamma_q$,
\begin{equation}
X^\xi_{\tau}(\om)=X_0-\sum_{i=1}^{q-1} z_{i,p_i} (t_{i+1}-t_i).
\end{equation} Because $V$ and $u$ are continuous (see Theorem \ref{cv}), $V$ is uniformly continuous on $C_N:=[t_1,T]\times \overline{B}(0,N)\times[-N,N]$ (where $\overline{B}(0,N)$ denotes the $d$-dimensional euclidian closed ball with radius $N$), and $u$ is uniformly continuous on $[-N,N]$. Therefore, we can find $\delta_N$ such that for every $t^i,x^i,r^i, i=1,2$, we have
 \[
|(t^1-t^2,x^1-x^2,r^1-r^2)|<\delta_N\Rightarrow |V(t^1,x^1,r^1)-V(t^2,x^2,r^2)|\vee |u(r^1)-u(r^2)|<\e.
\]
Further, take  $L\in\N$  such that 
\[
\frac{N}{2^L}<\delta_N,
\]
and introduce
\[
\G:=\big\{ ((1,p_1),\dots,(q,p_q))| q\in\{0,\dots,2^L\}, p_1,\dots,p_q\in\N\big\}.
\]
Setting
\begin{align*}
  r_j:=-N+\frac{jN}{2^{L}}, & \quad x_g:=X_0-\sum_{i=1}^{q-1} z_{i,p_i} (t_{i+1}-t_i),\\
  j\in\{1,...,2^{L+1}\},&\quad g\in\G, \text{ with }g=((1,p_1),\dots,(q,p_q)),
\end{align*} 
we can now define the following grid:
\[
\Gamma_N=\Big\{(t_i,x_g,r_l) | i\in\{0,...,2^{L}\},j\in\{0,...,2^{L+1}\}, g\in\G\Big\}\cap C_N.
\]
When
 \[
\Big(\tau(\om),X^\xi_\tau(\om), \cR^\xi_\tau(\om)\Big) \in \{t_i\}\times \{x_g\}\times[r_l,r_{l+1}[\,\cap\, C_N,
\]
we set
\[
\gamma_N( \om):=(T-t_{i},x_g,r_l).
\]
Note that $\gamma_N$ is $\F_\tau$-measurable. Let us denote by $\xi^{*,\gamma_N(\om)}$  the optimal strategy associated to $V(\gamma_N(\om))$ (which exists, due to Theorem \eqref{eos}).  Then,  the process $\xi^{*,\gamma_N(\om)}$ is well-defined for every $\om\in \big\{\big|X^\xi_\tau\big|\wedge\big|\cR^\xi_\tau\big|\leq N\big\}.$ Moreover, it belongs to the set $\dot{\X}^1_{2A_2}(T-t_{i}, x_g)=\dot{\X}^1_{2A_2}(T-\tau(\om),X_\tau^\xi(\om)).$ (Note that if $\tau(\om)=T$ and $x_g=0$, then $\gamma_N(\om)=(0,0,r_l)$, for some $r_l,$ which implies that $V(\gamma_N(\om))=u(r_l)$, and therefore $\xi^{*,\gamma_N(\om)}=0$ is well-defined in this case, too.) Furthermore, we have by construction
\begin{equation}
V(T-t_{i},x_g,r_l)= \E\Big[u\Big(r_l +\cR^{{\xi}^{*,\gamma_N(\om)}}_{\tau(\om),T}\Big)\Big],
\end{equation}
hence we obtain on $\big\{\big|X^\xi_\tau\big|\wedge\big|\cR^\xi_\tau\big|\leq N\big\}$:
\begin{eqnarray*}
\lefteqn{\Big|V(T-\tau(\om), X^\xi_{\tau}(\om),\cR^\xi_{\tau}(\om))-\E\Big[u\Big(\cR^\xi_\tau(\om) +\cR^{{\xi}^{*,\gamma_N(\om)}}_{\tau,T}(\om)\Big)\Big]\Big|}&\\
&&\leq\Big|V(T-\tau(\om), X^\xi_{\tau}(\om),\cR^\xi_{\tau}(\om))-V(\gamma_N(\om))\Big|\\
&&+\Big|V(\gamma_N(\om))-\E\Big[u\Big(\cR^\xi_\tau (\om)+\cR^{{\xi}^{*,\gamma_N(\om)}}_{\tau(\om),T}\Big)\Big]\Big|\\
&&=\Big|V(T-t_i,x_g,\cR^\xi_{\tau}(\om))-V(T-t_{i},x_g,r_l)\Big|\\
&&+\Big|\E\Big[u\Big(r_l +\cR^{{\xi}^{*,\gamma_N(\om)}}_{\tau(\om),T}\Big)\Big]-u\Big(\cR^\xi_\tau(\om) +\cR^{{\xi}^{*,\gamma_N(\om)}}_{\tau(\om),T}\Big)\Big]\Big|\\
&&\leq \e+\e\\
&&=2\e,
\end{eqnarray*}
due to the uniform continuity of $V$ and of $u$.
Thus, we have found  a process $\xi^{*,\gamma_N(.)}=\widetilde{\xi}^{.,\tau,\e}\in\dot{\X}^1_{2A_2}(T-\tau(.), X_\tau^\xi(.))$ such that \eqref{emi} holds for  every  $\om\in \big\{\big|X^\xi_\tau\big|\wedge\big|\cR^\xi_\tau\big|\leq N\big\}$.
 Moreover,  $$\widetilde{\xi}^{.,\tau,\e}\in\overline{K}_{m^{\e}}(T-\tau(.), X_\tau^\xi(.)),$$ where $m^{\e}$ has to be chosen as in \eqref{m&M}. \\
\\
\underline{Second step:} Let $\xi$ and $\tau$ be arbitrary. We can find a sequence of processes $\xi^k$ as in the first step such that $\xi^k$ converges to $\xi$ in $\L^p$, i.e., 
\[
\E\bigg[\int^T_0\big|\xi^{k}_t-\xi_t\big|^p\,dt\bigg]\longrightarrow 0,
\]
where $p$ is chosen according to Assumption \ref{afp}. Moreover, this sequence of processes may be chosen to lie in $\dot{\X}^1_{2A_2}(T,X_0)$, as argued in Assumption \ref{ar}. We will prove that
\begin{equation}
\cR^{\xi^k}_T\underset{k\rightarrow\infty}{\longrightarrow} \cR^{\xi}_T \quad\text{in probability}. \label{rkc}
\end{equation}
Due to Lemma \ref{xsc}, we have that 
\[
\int_t^T (X^{\xi^k}_s)^\top\sigma\,dB_s\underset{k\rightarrow \infty}{\longrightarrow }\int_t^T (X^{\xi}_s)^\top\sigma\,dB_s\quad \pas
\]
We have moreover, as a direct consequence of the $\L^p$ convergence of $\xi^k$ to $\xi$,
\[
 \int^T_t b \cdot X^{\xi^k}_s\,ds\underset{k\rightarrow \infty}{\longrightarrow } \int^T_t b \cdot X^{\xi}_s\,ds \quad \pas
\]
and
\[
\int^T_tf(-\xi^k_s)\,ds\underset{k\rightarrow \infty}{\longrightarrow }\int^T_tf(-\xi_s)\,ds\quad \text{in } \J
\]
(due to the growth condition imposed on $f$ in Assumption \ref{afp}), and hence in probability. This establishes \eqref{rkc}. \\
\\
\underline{Third step:} We can find a sequence of stopping times $(\tau_k)$ (with values in $[0,T[$) as in the first step such that $\tau_k\downarrow \tau\! \q$ As can be seen in the first step above, for each $k\in\N$, we can find $\widetilde{\xi}^{.,\tau_k,\e} \in\overline{K}_{m^{\e}}(T-\tau_k(.), X^{\xi^k}_{\tau_k}(.))$ such that 
\begin{equation}
V\big(T-\tau_k(\om),X^{\xi^k}_{\tau_k}(\om),\cR^{\xi^k}_{\tau_k}(\om)\big)\leq \E\Big[u\Big(\cR^{\xi^k}_{\tau_k}(\om) +\cR^{\widetilde{\xi}^{\om,\tau_k,\e}}_{\tau_k(\om),T}\Big)\Big]+\e \label{emik}
\end{equation}
for P-a.e  $\om\in \big\{\big|X^{\xi^k}_{\tau_k}\big|\wedge\big|\cR^{\xi^k}_{\tau_k}\big|\leq N\big\}$. Moreover,  we have
that $\widetilde{\xi}^{.,\tau_k,\e}\in\overline{\cK}_{m^{\e}}$, with 
\[
\overline{\cK}_{m^{\e}}= \Big\{\xi\in \overline{\C}\big(\dot{\X}^1_{2A_2}(T-\tau_k(.),X^\xi_{\tau_k}(.))\big)_k\big|\; \E\bigg[\int^T_{\tau(.)}f(-\xi_t)\,dt\bigg]\leq m^{\e}\Big\},
\]
where $  \overline{\C}(\dot{\X}^1_{2A_2}(T-\tau_k(.),X^\xi_{\tau_k}(.)))_k$ denotes the closed convex hull of the sequence of sets  $\big(\dot{\X}^1_{2A_2}(T-\tau_k(.),X^\xi_{\tau_k}(.))\big)_k$. Recall that we set here 
\[
\zeta_t=0 \;\text{for}\;  t\in [\tau(.),\tau_k(.)] \; \text{when}\; \zeta\in\dot{\X}^1_{2A_2}(T-\tau_k(.),X^\xi_{\tau_k}(.)),
\]
since $\tau(.)\leq \tau_k(.),\q\;\;$  \\
 Because $\overline{\cK}_{m^{\e}}$ is weakly sequentially compact, as proved in Proposition \ref{usc}, there exists $\widetilde{\xi}^{\tau,\e}\in\overline{\cK}_{m^{\e}}$ such that by passing to a subsequence if necessary, $\widetilde{\xi}^{k,\tau_k,\e}$ converges to $\widetilde{\xi}^{\tau,\e}$ weakly in $\J$. Using now Lemma \ref{wcx}, we have that $\widetilde{\xi}^{\tau,\e}\in\overline{\cK}_{m^{\e}}\q$ on $\{|X^{\xi^k}_{\tau_k}|\wedge|\cR^{\xi^k}_{\tau_k}|\leq N\}.$  
\\
\\
\underline{Last step:}
Notice first that we have
\begin{equation}
\limsup_k \E\Big[u\Big(\cR^{\xi^k}_{\tau_k}(\om) +\cR^{\widetilde{\xi}^{\om,\tau_k,\e}}_{\tau_k(\om),T}\Big)\Big]\leq \E\Big[u\Big(\cR^{\xi}_{\tau}(\om) +\cR^{\widetilde{\xi}^{\om,\tau,\e}}_{\tau(\om),T}\Big)\Big]  \label{lsem}
\end{equation}
for P-a.e $ \om\in \big\{\big|X^{\xi^k}_{\tau_k}\big|\wedge\big|\cR^{\xi^k}_{\tau_k}\big|\leq N\big\}$. Indeed, similarly to how it was established for $\xi\longmapsto \E\Big[u\big(\cR^\xi_T\big)\Big]$, we can prove that   $(r,\eta)\mapsto   \E\Big[u\big(r+\cR^\eta_{t,T}\big)\Big]$  is concave and thus we can apply Corollary \ref{lscc}, which proves \eqref{lsem}. (Note that we cannot simply apply Fatou's lemma to prove \eqref{lsem}, since it is not known whether or not
\[
\limsup_k u\Big(\cR^{\xi^k}_{\tau_k}(\om) +\cR^{\widetilde{\xi}^{\om,\tau_k,\e}}_{\tau_k(\om),T}\Big)\leq u\Big(\cR^{\xi}_{\tau}(\om) +\cR^{\widetilde{\xi}^{\om,\tau,\e}}_{\tau(\om),T}\Big), 
\]
because we only have a weak convergence of $\widetilde{\xi}^{\om,\tau_k,\e}$ to $\widetilde{\xi}^{\tau,\e}.$) Going back to \eqref{emik} and passing to the limit superior on both sides of the inequality, we finally get for P-a.e. $\om\in \{|X^{\xi}_{\tau}|\wedge|\cR^{\xi}_{\tau}|\leq N\}$,
\begin{align*}
 V\big(T-\tau(\om), X^{\xi}_{\tau}(\om), \cR^{\xi}_{\tau}(\om)\big)&=\limsup_k V\big(T-\tau_k(\om),X^{\xi^k}_{\tau_k}(\om),\cR^{\xi^k}_{\tau_k}(\om)\big)\\
&\leq\limsup_k \E\Big[u\Big(\cR^{\xi^k}_{\tau_k}(\om) +\cR^{\widetilde{\xi}^{\om,\tau_k,\e}}_{\tau_k(\om),T}\Big)\Big]+\e\\
&\leq \E\Big[u\Big(\cR^{\xi}_{\tau}(\om) +\cR^{\widetilde{\xi}^{\om,\tau,\e}}_{\tau(\om),T}\Big)\Big]+\e, 
\end{align*}
where the first equality is due to the continuity of $V$ in its arguments. This shows \eqref{emi}.
\end{proof}
We can now turn to proving Proposition \ref{eiv}
\begin{proof}[Proof of Proposition \ref{eiv}]
 Lemma \ref{eN} and Theorem \ref{em} imply for $\xi\in\dot{\X}^1_{2A_2}(T,X_0)$:
\begin{eqnarray*}
\lefteqn{\E[V(T-\tau, X^\xi_\tau, \cR^\xi_\tau)]}\\
&&=\E\Big[V(T-\tau, X^\xi_\tau, \cR^\xi_\tau)\b1_{\big\{\left|X^\xi_\tau\right|\vee\left|\cR^\xi_\tau\right|> N\big\}}\Big]+\E\Big[V(T-\tau, X^\xi_\tau, \cR^\xi_\tau)\b1_{\big\{\left|X^\xi_\tau\right|\wedge\left|\cR^\xi_\tau\right|\leq N\big\}}\Big]\\
&&\leq \e +  \int_\Om \E\Big[u\Big(\cR^\xi_\tau +\cR^{\widetilde{\xi}^{\om,\tau,\e}}_{\tau,T}\Big)\Big|\F_{\tau}\Big](\om)\P(d\om)+\e\\
&&=2\e+ \int_\Om\E\Big[u\big(\cR_{\tau}^{\xi}(\om)+\cR_{\tau(\om),T}^{\widetilde{\xi}^{\om,\tau,\e}}\big)\Big]\P(d\om)\\
&&= 2\e + \E\Big[u\Big(\cR^{\xi^{\tau,\e}}_T\Big)\Big]\\
&&\leq 2\e +V(T,X_0,R_0),
\end{eqnarray*}
due to Lemma \ref{ce}, whereby
the process $\xi^{\tau,\e}$ is defined as
\begin{equation*}
\xi^{\tau,\e}_t(\om)=\begin{cases}
   \xi_t(\om) & \text{ for } t\in [0,\tau(\om)]\\
 \widetilde{\xi}^{\om,\tau,\e}_t(\om) & \text{ for } t\in [\tau(\om),T],
 \end{cases}
\end{equation*}
 and the definition of $V(T,X_0,R_0)$.
   \end{proof}
In Proposition \ref{eiv} we have proved the inequality  $"\geq"$ of equation \eqref{ebp}. Now it remains to prove the reverse inequality. To this end, we need the following proposition, which uses the notion of the essential supremum of a set $\Phi$ of random variables, denoted by $\ess_\Phi$.
\begin{Prop}\label{esp}
With the  notations of Lemma \ref{ce}, we have 
\begin{equation}
   V\Big(T-\tau(\om), X_{\tau}^\xi(\om), \cR^{\xi}_{\tau}(\om)\Big)=\ess_{\xi^{\om}\in\dot{\X}^1_{2A_2}(T-\tau(\om), X_{\tau}^\xi(\om))}\E\left[u(\cR_{\tau}^\xi +\cR_{\tau, T}^{\xi^{\om}})|\F_{\tau}\right] (\om) \label{es}
\end{equation}
for $\pae\;\om\text{ on }\big\{\big|X^\xi_\tau\big|\wedge\big|\cR^\xi_\tau\big|\leq N\big\}$.
 \end{Prop}
\begin{proof}
We recall the $\pas$ equality fulfilled by $V(T-\tau, X^\xi_{\tau}, \cR_\tau^\xi)$,
\begin{equation*}
          V\left(T-\tau(\omega),X^\xi_{\tau}(\omega), \cR_\tau^\xi(\omega)\right)=\sup_{\xi^{\om}\in\dot{\X}^1_{2A_2}(T-\tau(\om),X^\xi_{\tau}(\omega))}\E\left[u\left(\cR^{\xi}_{\tau}(\omega)+\cR_{\tau,T}^{\xi^{\om}}(\omega)\right)\right]\q,
 \end{equation*}                                                                        
where $\xi^\om$ is defined as in Lemma \ref{ce}. Hence, this permits us to write
\begin{equation*}
V(T-\tau(\om), X^\xi_{\tau}(\om),\cR^{\xi}_{\tau}(\om))\geq\E\left[u\left(\cR_\tau^\xi+\cR_{\tau, T}^{\xi^{\om}}\right)\Big|\F_{\tau}\right](\om) \quad\text{P-a.s.~}
\end{equation*}
for all $\xi^{\om}\in\dot{\X}^1_{2A_2}(T-\tau(\om), X_{\tau}^\xi(\om))$. Using  the definition of the essential supremum (see, e.g., \citet{SF11},   Definition A.34), it follows then
\begin{equation}
 V(T-\tau(\om), X_{\tau}(\om),R^{X}_{\tau}(\om))\geq\ess_{\xi^{\om}\in\dot{\X}^1{2A_2}(T-\tau(\om), X_\tau^\xi(\om))}\E\left[u\left(\cR_\tau^\xi +\cR_{\tau, T}^{\xi^{\om}}\right)|\F_{\tau}\right](\om),
\label{famess}
\end{equation}
which proves the inequality $"\geq"$ of \eqref{es}. For the converse inequality, let $\widetilde{\xi}^{\om,\tau,\e}$ be as in Theorem \ref{em}. We have on $\{\big|X^\xi_\tau\big|\wedge\big|\cR^\xi_\tau\big|\leq N\big\}:$
\begin{equation*}
\E\left[u(\cR_\tau^\xi+\cR_{\tau, T}^{\widetilde{\xi}^{\om,\tau,\e}})|\F_{\tau}\right](\om)\geq V(T-\tau(\om), X^\xi_{\tau}(\om),\cR^{\xi}_{\tau}(\om))-\e\q
\end{equation*}
  And therefore
 \begin{equation*}
\ess_{\xi^\om\in\dot{\X}^1_{2A_2}(T-\tau(\om), X_\tau^\xi(\om))}\E\left[u(\cR_\tau^\xi+\cR_{\tau, T}^{\xi^\om})|\F_{\tau}\right](\om)\geq V(T-\tau(\om), X^\xi_{\tau}(\om),\cR^{\xi}_{\tau}(\om))-\e \q
\end{equation*}
 Letting $\e$ go to 0 gives us the required inequality.
\end{proof}
We can now prove Theorem~\ref{bp}.
\begin{proof}[Proof of Theorem \ref{bp}]
Thanks to Proposition \ref{eiv}, it remains to show only the inequality $"\leq"$ in  \eqref{ebp}. Let $\xi\in\dot{\X}^1_{2A_2}(T,X_{0})$ and set $\widetilde{\xi}_{s}=\xi_{\tau+t} \in\dot{\X}^1_{2A_2}(T-\tau, X^\xi_{\tau}) $ for $s\geq\tau$ and $t\geq0$. The definition of the essential supremum, in conjunction with  Proposition~\ref{esp} and Lemma \ref{eN}, yields
\begin{align*}
\E\Big[u\big(\cR^\xi_T\big)\Big]&=\E\Big[u\Big(\cR^\xi_\tau +\cR^{\widetilde{\xi}}_{\tau,T}\Big)\Big]\\
&=\E\left[\E\left[u(\cR_\tau^\xi+\cR_{\tau, T}^{\widetilde{\xi}})|\F_{\tau}\right]\right]\\
&=\E\left[\E\left[u(\cR_\tau^\xi+\cR_{\tau, T}^{\widetilde{\xi}})|\F_{\tau}\right]\Big(\b1_{\big\{\left|X^\xi_\tau\right|\vee\left|\cR^\xi_\tau\right|> N\big\}}+\b1_{\big\{\left|X^\xi_\tau\right|\wedge\left|\cR^\xi_\tau\right|\leq N\big\}}\Big)\right]\\
&\leq \e+\E\Big[V(T-\tau, X^\xi_\tau, \cR^\xi_\tau)\b1_{\big\{\left|X^\xi_\tau\right|\wedge\left|\cR^\xi_\tau\right|\leq N\big\}}\Big].
\end{align*}
Taking the supremum over $\xi$ 
and then sending $\e$ to zero (which implies sending $N$ to infinity), 
shows the assertion. 
\end{proof}

\bibliographystyle{plainnat}

\bibliography{A03n,AC01n,AL07,B11,B13,BB04,BJ02,BK04,BL98n,BN12n,BS91n,BT11n,CC14,CIL92n,CJ04,CST07,D74,DS88n,FShm06,FT11,HK04,HS04,J94,KK11,KL11,K70n,K08,K85,M66,O06n,P09n,PF03n,Pr04,R97n,RU78,S08,SB78,SSB08,SF11n,SS09,SST09n,SST10n,SS07,T04n,T12,Ta11,W13,W41,W80,W91,YR91,YZ99n,ZB04,ZCB12}

\end{document}